\documentclass[14pt, a4paper, reqno]{amsart}
\usepackage{mathrsfs}
\usepackage{bbm}
\usepackage{amsmath,amscd,amssymb,latexsym}
\usepackage[all]{xy}

\input xypic

\evensidemargin=\oddsidemargin
\DeclareMathVersion{can}
\DeclareMathAlphabet{\can}{OT1}{cmss}{m}{n}
\newtheorem{thm}{Theorem}[section]

\newtheorem{lem}[thm]{Lemma}

\newtheorem{exa}[thm]{Example}
\theoremstyle{definition}

\theoremstyle{fact}

\theoremstyle{conjecture}

\numberwithin{equation}{section}

\newcommand{\ord}{\operatorname{ord}}

\begin{document}
\title[Further   factorization of $x^{n}-1$ over finite fields (II)] {Further  factorization of $x^n-1$ over  finite fields (II)}


\author[Y. Wu]{Yansheng Wu}
\address{\rm School of Computer Science, Nanjing University of Posts and Telecommunications, Nanjing 210023, P. R. China  }
 \email{yanshengwu@njupt.edu.cn}

\author[Q. Yue]{Qin Yue}
\address{\rm Department of Mathematics, Nanjing University of Aeronautics and Astronautics,
Nanjing, 210016, P. R. China; State Key Laboratory of Cryptology, P. O. Box 5159, Beijing, 100878, P. R. China}
\email{yueqin@nuaa.edu.cn}


\subjclass[2010]{11T06, 12E05, 94B15}

\keywords{ Irreducible factor, cyclotomic polynomial}

\begin{abstract} Let $\Bbb F_q$ be a finite field with $q$ elements.
Let  $n$ be a positive integer with radical $rad(n)$, namely, the product of distinct prime divisors of $n$.
If the order of $q$ modulo $rad(n)$ is either 1 or a prime, then
   the irreducible factorization and a counting formula of irreducible factors of $x^n-1$ over $\Bbb F_q$ were obtained by  Mart\'{\i}nez, Vergara, and Oliveira (Des Codes Cryptogr  77 (1) : 277-286, 2015) and Wu, Yue, and Fan (Finite Fields Appl 54: 197-215, 2018).
In this paper,  we explicitly factorize  $x^{n}-1$ into irreducible factors in $\Bbb F_q[x]$ and calculate  the number of the irreducible factors when  the order of  $q$ modulo $rad(n)$ is a product of two primes.
\end{abstract}
\maketitle

\section{Introduction}

There is a close connection between the factorization of a polynomial over a finite field and  a wide variety of technological situations, including efficient and secure communications, error-correcting codes, deterministic simulations of random processes and digital tracking systems (see for instance \cite{G}).
Moreover, in theory, the factorization of  polynomials over finite fields is a classical topic of mathematics. There are various computational problems depending in one way or another on the factorization of polynomials over finite fields. The factorization of $x^{n} - 1$ is inextricably linked with the factorization of the $n$-th cyclotomic polynomial $\Phi_n(x)$ (see \cite{LN}), which has received a good deal of attention.

The explicit factorization of  $\Phi_{2^{n}r}(x )$  over $\Bbb F_{q}$ was studied in \cite{F},  where $r$ is an odd prime and $q \equiv \pm1\pmod r$.
This gave the explicit irreducible factors of $\Phi_{2^{n}3}(x )$ and the Dickson polynomial $D_{2^{n}3}(x )$ completely.
The paper \cite{WW} showed that  all irreducible factors of $\Phi_{2^{n}r}(x )$ can be obtained easily from irreducible factors of cyclotomic polynomials of small orders. Hence, the explicit factorization of $\Phi_{2^{n}5}(x )$ over finite fields was obtained by this way.
Assuming that the explicit factors of  $\Phi_{r}(x )$ are known,  Tuxanidy and  Wang obtained the irreducible factors of  $\Phi_{2^{n}r}(x )$ over $\Bbb F_{q}$ concretely in \cite{TW}, where $r \geqslant3$ is an arbitrary odd integer.
 Chen {\em et al.} \cite{CLT} obtained the irreducible factorization of $x^{2^mp^n}-1$ over $\Bbb F_q$  in a very explicit form, where $p$ is an odd prime divisor of $q-1$. It is also shown that all the irreducible factors of $x^{2^mp^n}-1$ over $\Bbb F_q$ are either binomials or trinomials.

Most recently,   Wu {\em et al.} \cite{WZF}  showed that all irreducible factors of $\Phi_{p^{n}}(x )$ and $\Phi_{p^{n}r}(x )$ can be obtained from irreducible factors of cyclotomic polynomials of small orders. Applying this result, they presented a general idea to factorize cyclotomic polynomials over finite fields.
Some explicit factorizations of certain cyclotomic polynomials or Dickson polynomials can be found in \cite{BGM}-\cite{LN}, \cite{W,WY},  \cite{WYF, YKT, YY, YYZ}.

Suppose that the order of $q$ modulo $rad(n)$ is  1,   Mart\'{\i}nez {\em et al.}  \cite{MVO} explicitly factorized   $x^{n} -1$ into irreducible binomials $x^{t} -a$ or trinomials $x^{2t} - ax^{t} + b$ over $\Bbb F_q$ and counted the number of irreducible polynomials.  Suppose that the order of $q$ modulo $rad(n)$ is a prime,   the irreducible factorization and a counting formula of   irreducible factors of $x^n-1$ over $\Bbb F_q$ were given in  \cite{WYF}.

In this paper,  we explicitly factorize  $x^{n}-1$ into irreducible factors in $\Bbb F_q[x]$ and calculate the number of the irreducible factors when  the order of  $q$ modulo $rad(n)$ is a product of two primes. By comparing with the case of the order of $q$ modulo $rad(n)$ is a prime number \cite{WYF}, the results presented in this paper is more complicated due to the number of conjugate polynomials of an irreducible factor  has to face  many situations.

This paper is organized as follows.  Section 2 presents some basic results.   Section 3  gives the  irreducible factorization of  $x^{n}-1$ over $\Bbb F_{q}$ and a counting formula of  its irreducible factors under the condition the order of $q$ modulo $rad(n)$ is a square of a prime.
In Section 4,  we explicitly factorize $x^{n}-1$ into irreducible factors over $\Bbb F_{q}$ and calculate the number of its  irreducible factors when the order of $q$ modulo $rad(n)$ is a product of two distinct primes. 
 Section 5
 concludes this paper.



 \section{Preliminaries}


Let $n$ be a  positive integer and  $ n= p_{1}^{\alpha_{1}}p_{2}^{\alpha_{2}} \cdots p_{l}^{\alpha_{l}}$,  where $p_1,\ldots, p_l$ are distinct primes and $\alpha_i\ge 1$, $1\le i\le l$. We denote $rad(n) = p_{1} p_{2} \cdots p_{l}$ and  $v_{p_i}(n)=\alpha_i$, $1\le i\le l$.
 For  two integers $m$ and $n$ with $\gcd(m,n)=1$, $\ord_{n} (m)$  denotes the order of $m$ modulo $n$, i.e.  it is the minimum positive integer $k$ such that $m^k \equiv 1 \pmod n$.

Let $\Bbb F_q$ be a finite field of order $q$, where $q$ is a positive power of a prime $p$ and $o(\alpha)$ denote the order of $\alpha\in \Bbb F^{\ast}_q$ in the multiplicative group $\Bbb F^{\ast}_q$. Through this paper,  we always assume that $\gcd(n,q)=1$.

There is a classical remarkable criterion on irreducible binomials over $\Bbb F_q$.
\begin{lem} \cite[Theorem 3.75]{LN} {\rm 
Assume that $t\geq 2$ is a positive integer and $\eta \in \Bbb F_q^\ast$. Then the binomial
 $x^t-\eta$ is irreducible over $\Bbb F_q$ if and only if all the following conditions are satisfied:

  (i) $ rad(t)$ divides $o(\eta)$,

  (ii) $\gcd(t, \frac{q-1}{o(\eta)})=1$,

(iii) if $4| t$, then $4 | (q-1)$.}
\end{lem}

In fact, if $t=4s$, $q\equiv 3\pmod 4$ and $a$ is not a square in $\Bbb F_q$, then $x^t-a=(x^{2s}-bx^s+c)(x^{2s}+bx^s+c)$ is reducible over $\Bbb F_q$, where $b, c\in \Bbb F_q$,  $c^2=-a$, and $b^2=2c$.

The following  lemmas are  concerned with the explicit factorization of  $x^{n}-1$ over $ \Bbb F_{q}$.

\begin{lem}\cite[Corollary 1]{MVO} {\rm Let $ \Bbb F_q$ be a finite field and $n$ a positive integer such that both $rad(n)|(q-1)$ and
either
 $q \not\equiv 3 \pmod {4}$ or $8\nmid n$. Let $m_1=\frac {n}{\gcd(n, q-1)}$, $l_1=\frac{q-1}{\gcd(n, q-1)}$, and  $\theta$  a  generator of $\Bbb F_q^*$. Then
the factorization of $x^{n} -1$ into irreducible factors in $\Bbb F_q [x ]$ is
$$\prod\limits_{t\mid m_{1}} \prod_{\mbox{\tiny$
\begin{array}{c}
1\leqslant u\leqslant \gcd(n,q-1)\\
\gcd(u,t)=1\\
\end{array}$
}}(x^{t}-\theta^{ul_{1}}).$$

Moreover,
for each $t| m_{1}$, the number of irreducible factors of degree $t$ is $\frac{\varphi(t)}{t}\cdot \gcd(n,q-1)$, where $\varphi$ denotes the Euler Totient function,  and the number of irreducible factors is
$$\gcd(n,q-1)\cdot \prod_{\mbox{\tiny$
\begin{array}{c}
p| m_{1}\\
p~ prime\\
\end{array}$
}}\bigg(1+v_{p}(m_{1})\cdot \frac{p-1}{p}\bigg).$$ }
\end{lem}

\begin{lem}\cite[Corollary 2]{MVO} {\rm Let $ \Bbb F_q$ be a finite field and $n$ a positive integer such that  $rad(n)|(q-1)$,
 $q \equiv3 \pmod 4$, and   $8|n$. Let $m_2=\frac{n}{\gcd(n, q^2-1)}$, $l_2=\frac{q^2-1}{\gcd(n, q^2-1)}$, $l_1=\frac{q-1}{\gcd(n, q-1)}$, $r=\min\{v_2(n/2),v_2(q+1)\}$, and $\alpha$ a  generator of $\Bbb F_{q^2}^*$ satisfying $\theta=\alpha^{q+1}$.
Then the factorization of $x^{n} -1$ into irreducible factors in $\Bbb F_q [x ]$ is
$$\prod_{\mbox{\tiny$
\begin{array}{c}
t| m_{2}\\
t~odd\\
\end{array}$
}} \prod_{\mbox{\tiny$
\begin{array}{c}
1\leqslant v\leqslant \gcd(n,q-1)\\
\gcd(v,t)=1\\
\end{array}$
}}(x^{t}-\theta^{vl_{1}})\cdot\prod\limits_{t\mid m_{2}} \prod\limits_{u\in \mathcal{R}_{t}} (x^{2t}-(\alpha^{ul_{2}}+\alpha^{qul_{2}})x^{t}+\theta^{ul_{2}}), $$ where
$$\mathcal{R}_t=\bigg\{u\in \Bbb N: \begin{array}{l}
1\le u\le \gcd(n,q^{2}-1),2^{r}\nmid u,
 \\  \gcd(u,t)=1, u<\{qu\}_{\gcd(n,q^{2}-1)}\end{array}\bigg\}.$$
Note that  $\{a\}_{b}$ denotes the remainder of the division of $a$ by $b$.

Moreover,  for each odd $t$  with $t| m_{2}$, the number of irreducible polynomials of degree $t$  is $\frac{\varphi(t)}{t}\cdot \gcd(n,q-1)$; the number of irreducible polynomials of degree $2t$ is
\begin{equation*}\left\{
   \begin{array}{ll}
    \frac{\varphi(t)}{2t}\cdot 2^{r_{}}\gcd(n,q-1), &\mbox{ if $ t$ is even},\\
       \frac{\varphi(t)}{2t}\cdot (2^{r_{}}-1)\gcd(n,q-1), &\mbox{ if $ t$ is odd}.\\
    \end{array}\right.
  \end{equation*}
The number of irreducible factors of $x^{n}-1$ over $\Bbb F_q$ is
$$\gcd(n,q-1)\cdot \bigg(\frac1 2 +2^{r_{}-2}(2+v_{2}(m_{2}))\bigg)\cdot \prod_{\mbox{\tiny$
\begin{array}{c}
p|m_{2}\\
p~ odd ~prime\\
\end{array}$
}}\bigg(1+v_{p}(m_{2})\frac{p-1}{p}\bigg).$$}
\end{lem}

\begin{lem}\cite[Lemma 14]{WZF}{\rm Let $n$ be a positive integer, $q$  a prime power coprime with $n$, and $k$  the smallest positive integer such that $q^{k} \equiv 1 \pmod {rad(n)}$. If $f (x)$ is an irreducible factor of $\Phi_{n} (x)$ over $\Bbb F_{q^{k}}$, then $f^{\sigma}(x)$ is also an irreducible factor of $\Phi_{n} (x)$ over $\Bbb F_{q^{k}}$, where $\sigma : \Bbb F_{q^{k}} \longrightarrow
\Bbb F_{q^{k}} ; \alpha \longmapsto \alpha ^{q}$, is the Frobenius mapping of $\Bbb F_{q^{k}} $ over $\Bbb F_{q} $. Moreover, $f(x)f^{\sigma}(x)f^{\sigma^{2}}(x)\cdots f^{\sigma^{k-1}}(x)$ is an irreducible factor of  $\Phi_{n} (x)$ over $\Bbb F_{q} $.}
\end{lem}
It is easy to verify the following result.

\begin{lem}{\rm If $n$ is  a positive integer and $a,b$ are two positive divisors of $n$, then $\gcd(\frac n a, \frac n b)=\frac n{{lcm}(a,b)}$, where $\mbox{lcm}(a,b)$ denotes the least common multiple of $a$ and $b$.}
\end{lem}

\section{Case: $w=w_1=w_2$}





In this section,  we always assume that  $w$ is  a prime and $\ord_{rad(n)}(q)=w^2$.


Suppose that $x^{n}-1$ has an irreducible factorization over $\Bbb F_{q^{w^2}}$ as follows:
\begin{equation}x^{n}-1=\prod\limits_{d|n}\Phi_{d}(x)=\prod\limits_{d|n}f_{d,1}(x)f_{d,2}(x)\cdots f_{d,s_{d}}(x).\end{equation}

 For each positive  divisor $d$ of $n$,  the order of $q$ modulo $rad(d)$ is $1$,  $w$, or $w^2$ by $rad(d)| rad(n)$ and $w$ a prime. We define two sets  $$D_1=\{d\mid n: 1\le  d\le n, rad(d)|(q-1)\},D_w=\{d\mid n: 1\le  d\le n, d\notin D_1, rad(d)|(q^w-1)\}.$$ Those irreducible polynomials in (3.1) can be divided into three parts:
$$X_{(0)}=\{f_{d,j}(x):d\in D_1, 1\leqslant j\leqslant s_{d}\},X_{(1)}=\{f_{d,j}(x):d\in D_w, 1\leqslant j\leqslant s_{d}\},$$
$$ X_{(2)}=\{f_{d,j}(x):d\notin D_1\cup D_w, 1\leqslant j\leqslant s_{d}\}.$$
 From \cite{WZF}, there are equivalence relations $\sim_1$ and  $\sim_2$ on $X_{(1)}$ and $ X_{(2)}$, respectively:   $$f_{d,i}(x)\sim_1 f_{d,j}(x)~ {if}~ and ~only~ if~ f_{d,j}(x)= f_{d,i}^{\sigma^{k}}(x) ~for ~some~k, 0 \leqslant k\leqslant w - 1,$$  $$f_{d,i}(x)\sim_2 f_{d,j}(x)~ {if}~ and ~only~ if~ f_{d,j}(x)= f_{d,i}^{\sigma^{k}}(x) ~for ~some~k, 0 \leqslant k\leqslant w^2 - 1. $$
By  Lemma 2.4, we have the irreducible factorization of $x^{n}-1$ over $\Bbb F_{q}$:
$$\prod\limits_{g(x)\in X_{(0)} }g(x) \prod\limits_{f(x)\in X^{\ast}_{(1)} } \prod\limits_{k=0 }^{w-1}f^{\sigma^k}(x)\prod\limits_{h(x)\in X^{\ast}_{(2)} } \prod\limits_{k=0 }^{w^2-1}h^{\sigma^k}(x),$$
 where $X^{\ast}_{(1)}$ and $X^{\ast}_{(2)}$ are complete systems  of equivalence class representatives of $X_{(1)}$ and $X_{(2)}$ relative to $\sim_1$ and  $\sim_2$, respectively.

 In the following, we consider the irreducible factorization of $x^n-1$ over $\Bbb F_q$ by two cases: (1) $w$ is an odd prime, (2) $w$=2.

  \subsection{ $w$  is an odd prime} $~$

  In the subsection, we always assume  that $w$ is  an odd prime and $\ord_{rad(n)}(q)=w^2$.
 
\begin{lem}{\rm Let $w$ be an odd prime.

$(1)$ Then   $$\gcd(q-1, \frac{q^{w}-1}{q-1})=\gcd(q-1, w)=\left\{\begin{array}{ll}
1, &\mbox{ if $w\nmid (q-1)$,}\\
w, &\mbox{ if $w|(q-1)$,}\\
\end{array}\right. $$ and $$\gcd(q-1, \frac{q^{w^2}-1}{q-1})=\gcd(q-1,w^2)=\left\{\begin{array}{ll}
1, &\mbox{ if $w\nmid (q-1)$,}\\
w, &\mbox{ if $v_w(q-1)=1$,}\\
w^2, &\mbox{ if $v_w(q-1)\geqslant 2$.}\\
\end{array}\right.$$

$(2)$ If  $w|(q-1)$, then  $v_w(\frac{q^{w}-1}{q-1})=1$ and $v_w(\frac{q^{w^2}-1}{q-1})=2$.}
\end{lem}

 \begin{proof} (1) Since $\frac{q^{w}-1}{q-1}=1+q+q^{2}+\cdots+q^{w-1}=w+(q-1)+(q^{2}-1)+\cdots+(q^{w-1}-1)$, $\gcd(q-1,\frac{q^{w}-1}{q-1})=\gcd(q-1, w)=1$ or $w$.
 Since $\frac{q^{w^2}-1}{q-1}=1+q+q^{2}+\cdots+q^{w^2-1}=w^2+(q-1)+(q^{2}-1)+\cdots+(q^{w^2-1}-1)$, $\gcd(q-1,\frac{q^{w^2}-1}{q-1})=\gcd(q-1, w^2)=1$ or $w$ or $w^2$.

 (2) If $w$ is an odd prime and $v_w(q-1)=r\ge 1$, i.e. $q=1+aw^r$ and $\gcd (a, w)=1$, then $q^{w}-1\equiv aw^{r+1} \pmod {w^{r+2}}$ and $q^{w^2}-1\equiv aw^{r+2} \pmod {w^{r+3}}$. Thus
 $v_w(\frac{q^w-1}{q-1})=1$   and $v_w(\frac{q^{w^2}-1}{q-1})=2$.
 \end{proof}

 \begin{lem} {\rm Let $n=w^{v_w(n)}n'$, where $w$ is odd and $\gcd(n', w)=1$.

 $(1)$ If $v_w(n)=0$ or $v_w(n)\geqslant v_w(q^{w^2}-1)$, then
$\gcd(n, q^{w^2}-1)=\gcd(n,q-1)\gcd(n, \frac{q^{w^2}-1}{q-1})=\gcd(n,q^w-1)\gcd(n, \frac{q^{w^2}-1}{q^w-1})$, $lcm(q-1,\gcd(n, q^{w^2}-1) )=(q-1)\gcd(n, \frac{q^{w^2}-1}{q-1})$, and $lcm(q^w-1,\gcd(n, q^{w^2}-1) )=(q^w-1)\gcd(n, \frac{q^{w^2}-1}{q^w-1})$.

  $(2)$ If  $1=v_w(n)<v_w(q^{w^2}-1)$, then
$\gcd(n, q^{w^2}-1)=\gcd(n/w,q-1)\gcd(n, \frac{q^{w^2}-1}{q-1})=\gcd(n/w,q^w-1)\gcd(n, \frac{q^{w^2}-1}{q^w-1}) $, $lcm(q-1,\gcd(n, q^{w^2}-1) )=\frac{q-1}{w}\gcd(n, \frac{q^{w^2}-1}{q-1})$, and $lcm(q^w-1,\gcd(n, q^{w^2}-1) )=\frac{q^w-1}{w}\gcd(n, \frac{q^{w^2}-1}{q^w-1})$.

   $(3)$ If  $2\leqslant v_w(n)<v_w(q^{w^2}-1)$, then
$\gcd(n, q^{w^2}-1)=\gcd(n/w^2,q-1)\gcd(n, \frac{q^{w^2}-1}{q-1})$ and $lcm(q-1,\gcd(n, q^{w^2}-1) )=\frac{q-1}{w^2}\gcd(n, \frac{q^{w^2}-1}{q-1})$.}
 \end{lem}

\begin{proof} We only need to prove  (2) and (3). If $w|n$, then $w|(q^{}-1)$ and $w|(q^{w^2}-1)$. By Lemma 3.1, $v_w(\frac{q^{w}-1}{q-1})=1$ and $v_w(\frac{q^{w^2}-1}{q-1})=2$. If  $1=v_w(n)<v_w(q^{w^2}-1)$, then $w\nmid \gcd(n/w,q-1)$ and $w\nmid \gcd(n/w,q^w-1)$. If  $2\leqslant v_w(n)<v_w(q^{w^2}-1)$, then $v_w(\gcd(n, \frac{q^{w^2}-1}{q-1}))=2$ and {$v_w(\gcd(n, \frac{q^{w^2}-1}{q^w-1}))=1$}.
\end{proof}

First, we consider the case:  either $q\not\equiv3\pmod 4$ or $8\nmid n$.
 Let $m_{w^2}=\frac{n}{\gcd(n,q^{w^2}-1)}$, $l_{w^2}=\frac{q^{w^2}-1}{\gcd(n,q^{w^2}-1)}$, and  $\delta$ a  generator of $\Bbb F_{q^{w^2}}^*$. By Lemma 2.2, there is an irreducible factorization of $x^n-1$ over $\Bbb F_{q^{{w^2}}}$:
 \begin{equation}
 x^{n} -1=\prod\limits_{t|m_{w^2}} \prod_{\mbox{\tiny$
\begin{array}{c}
1\leqslant u\leqslant \gcd(n,q^{{w^2}}-1)\\
\gcd(u,t)=1\\
\end{array}$
}}(x^{t}-\delta^{ul_{w^2}}).
\end{equation}

  Next we will explicitly give the irreducible factorization of $x^{n}-1$ over $\Bbb F_{q}$. In fact, we investigate whether $\delta^{ul_{w^2}}\in \Bbb F_q$ and $\delta^{ul_{w^2}}\in \Bbb F_{q^w}\verb|\|  \Bbb F_{q}$.

\begin{thm} {\rm Suppose that $w$ is an odd prime, $\ord_{rad(n)}(q)={w^2}$, and either $q\not\equiv3 \pmod 4$ or $8\nmid n$.
 Let $n=w^{v_w(n)}n_1n_2n_3$, $rad(n_1)|(q-1)$, $rad(n_2)|\frac{q^{w}-1}{q-1}$, and $rad(n_3)|\frac{q^{w^2}-1}{q^w-1}$. Let $m_{w^2}=\frac{n}{\gcd(n,q^{w^2}-1)}$,  $m_{{w^2},1}=\frac{n_1}{\gcd(n_1, q-1)}$, $m_{{w^2},2}=\frac{n_1n_2}{\gcd(n_1n_2, q^w-1)}$, $l_{w^2}=\frac{q^{w^2}-1}{\gcd(n,q^{w^2}-1)}$, $l_{w}=\frac{q^{w}-1}{\gcd(n,q^{w}-1)}$, $l_{1}=\frac{q-1}{\gcd(n, q-1)}$,  $\delta$ a  generator of $\Bbb F^{\ast}_{q^{w^2}}$ satisfying  $\theta=\delta^{\frac{q^{w^2}-1}{q-1}}$, and $\pi$ a  generator of $\Bbb F^{\ast}_{q^{w}}$ satisfying  $\pi=\delta^{\frac{q^{w^2}-1}{q^w-1}}$. Then

$(1)$ The irreducible factorization of  $x^{n}-1$  over $\Bbb F_{q}$  is
\begin{eqnarray*}&&\prod_{\mbox{\tiny$
\begin{array}{c}
 t|m_{w^2,1}\\
\end{array}$
}}
 \prod_{\mbox{\tiny$
\begin{array}{c}
1\le u\le \gcd(n, q-1)\\
\gcd(u, t)=1
\end{array}$
}}(x^{t}-\theta^{ul_{1}})\\
&\cdot&\prod_{\mbox{\tiny$
\begin{array}{c}
t|m_{w^2,2}\\
\end{array}$
}}\prod_{\mbox{\tiny$
\begin{array}{c}
u_1\in \mathcal{S}^{(1)}_{t}
\end{array}$
}}
\prod\limits_{k=0}^{w-1} (x^{t}-\pi^{q^ku_1l_{w}})\cdot\prod_{\mbox{\tiny$
\begin{array}{c}
t|m_{w^2}\\
\end{array}$
}}\prod_{\mbox{\tiny$
\begin{array}{c}
u_2\in \mathcal{S}^{(2)}_{t}
\end{array}$
}}
\prod\limits_{k=0}^{w^2-1} (x^{t}-\delta^{q^ku_2l_{w^2}}),
\end{eqnarray*}
where  $$\mathcal{S}^{(1)}_{t}=\bigg\{u_1\in \Bbb N: \begin{array}{l}
1\le u_1\le \gcd(n,q^{w}-1),\gcd(u_1,t)=1,
 \\ \frac{q^{w}-1}{q-1}\nmid u_1l_{w}, u_1=\min\{u_1, qu_1, \cdots, q^{w-1}u_1\}_{\gcd(n, q^w-1)}\end{array}\bigg\}$$
and  $$\mathcal{S}^{(2)}_{t}=\bigg\{u_2\in \Bbb N: \begin{array}{l}
1\le u_2\le \gcd(n,q^{w^2}-1),\gcd(u_2,t)=1,
 \\ \frac{q^{w^2}-1}{q^w-1}\nmid u_2l_{w^2}, u_2=\min\{u_2, qu_2, \cdots, q^{w^2-1}u_2\}_{\gcd(n, q^{w^2}-1)}\end{array}\bigg\}.$$

$(2)$  The number of irreducible factors is
\begin{eqnarray*}&&\prod_{\mbox{\tiny$
\begin{array}{c}
p| m_{w^2,1}\\
p~prime\\
\end{array}$
}}(1+v_p(m_{w^2,1})\frac{p-1}p)\cdot \frac{w-1}{w}\cdot\gcd(n, q-1) \\&+&\prod_{\mbox{\tiny$
\begin{array}{c}
p| m_{w^2,2}\\
p~prime\\
\end{array}$
}}(1+v_p(m_{w^2,2})\frac{p-1}p)\cdot \frac{w-1}{w^2}\cdot\gcd(n, q^w-1) \\ &+&\prod_{\mbox{\tiny$
\begin{array}{c}
p| m_{w^2}\\
p~prime\\
\end{array}$
}}(1+v_p(m_{w^2})\frac {p-1}p) \cdot\frac{1}{w^2}\cdot\gcd(n, q^{w^2}-1).
\end{eqnarray*} }
\end{thm}

\begin{proof} Let $n=w^{v_w(n)}n_1n_2n_3$, $rad(n_1)|(q-1)$, $rad(n_2)|\frac{q^{w}-1}{q-1}$, and $rad(n_3)|\frac{q^{w^2}-1}{q^w-1}$. Let $m_{{w^2},1}=\frac{n_1}{\gcd(n_1, q-1)}$ and
$$m_{{w^2},2}=\frac{n_{1}n_2}{\gcd(n_1n_2,q^w-1)}=m_{{w^2},1}\cdot \frac{n_{2}}{\gcd(n_2,\frac{q^w-1}{q-1})}.$$
Then $$m_{w^2}=\frac{n}{\gcd(n,q^{w^2}-1)}=w^{v_w(n)-v_w(\gcd(n, q^{w^2}-1))}\cdot m_{w^2,2}\cdot  \frac{n_{3}}{\gcd(n_3,\frac{q^{w^2}-1}{q^w-1})} .$$

In (3.2), by  Lemma 2.2,  for each divisor $t$ of $m_{w^2}$ the number of irreducible polynomials of degree $t$ in $\Bbb F_{q^{w^2}}[x]$ is
$\frac{\varphi(t)}t\cdot\gcd(n, q^{w^2}-1)$, and   $x^t-\delta^{ul_{w^2}}\in \Bbb F_q[x]$ if and only if $\delta^{ul_{w^2}}\in \Bbb F_q$, namely, $\frac{q^{w^2}-1}{q-1}| ul_{w^2}$.  Moreover, if  $x^t-\delta^{ul_{w^2}}\notin \Bbb F_q[x]$, then    $x^t-\delta^{ul_{w^2}}\in\Bbb F_{q^w}[x]  $ if and only if $\delta^{ul_{w^2}}\in \Bbb F_{q^w} \verb|\|  \Bbb F_{q}$, which is equivalent to $\frac{q^{w^2}-1}{q^w-1}| ul_{w^2}$ and $\frac{q^{w^2}-1}{q-1}\nmid ul_{w^2}$.

In the following, we find these $u$, $1\le u\le \gcd(n, q^{w^2}-1)$ and  $\gcd(u,t)=1$,  such that either $\delta^{ul_{w^2}}\in \Bbb F_q$ or   $\delta^{ul_{w^2}}\in \Bbb F_{q^w} \verb|\|  \Bbb F_{q}$.
We proceed the proof in three cases of  $w$-exponential valuation of $n$.

{\bf Case 1.} $v_w(n)=0$ or $v_w(n)\geqslant v_w(q^{w^2}-1)$.  By Lemmas 2.5, 3.1, and 3.2,
$$\gcd(\frac{q^{w^2}-1}{q-1}, l_{w^2})=\frac{{q^{w^2}-1}}{lcm(q-1, \gcd(n, q^{w^2}-1))}=\frac{\frac{q^{w^2}-1}{q-1}}{\gcd(n, \frac{q^{w^2}-1}{q-1})}.$$
Then $x^t-\delta^{ul_{w^2}}\in \Bbb F_q[x]$   if and only if $\gcd(n, \frac{q^{w^2}-1}{q-1})| u$.

 Suppose that $t| m_{{w^2},1}$. Then   $x^t-\delta^{ul_{w^2}}\in\Bbb F_q[x]$ if and only if $u=\gcd(n, \frac{q^{w^2}-1}{q-1})u'$, $1\le u'\le \gcd(n, q-1)$ and $\gcd(u',t)=1$.  Let $p_{1},p_{2},\cdots,p_{k}$ be all distinct prime divisors of $t$. Then there exist $(1 - \frac{1}{p_{1}} ) \cdot \gcd(n,q -1)$ numbers of  $u'$, $1\le u'\le \gcd(n, q-1)$, such that $\gcd(u',p_{1})=1$.
Inductively we conclude that $(1-\frac{1}{p_{1}})(1-\frac{1}{p_{2}})\cdots (1-\frac{1}{p_{k}})\cdot \gcd(n,q-1)=\frac{\varphi(t)}{t}\cdot\gcd(n,q-1)$ numbers of $u'$ such that $\gcd(u',t)=1$.
Hence
there exist $\frac{\varphi(t)}t\cdot
\gcd(n, q-1)$   irreducible polynomials  of degree $t$ in $\Bbb F_{q^{w^2}}[x]$ that are in $\Bbb F_q[x]$. Moreover, there is no   polynomial  $x^t-\delta^{ul_{w^2}}\in \Bbb F_q[x]$ if  $t\nmid m_{w^2,1}$.

By Lemmas 2.5, 3.1, and 3.2,
$$  \gcd(\frac{q^{w^2}-1}{q^w-1}, l_{w^2})=\frac{{q^{w^2}-1}}{lcm(q^w-1, \gcd(n, q^{w^2}-1))}=\frac{\frac{q^{w^2}-1}{q^w-1}}{\gcd(n, \frac{q^{w^2}-1}{q^w-1})},$$
   $x^t-\delta^{ul_{w^2}}\notin \Bbb F_q[x]$
and $x^t-\delta^{ul_{w^2}}\in \Bbb F_{q^w}[x]$    if and only if $\gcd(n, \frac{q^{w^2}-1}{q^w-1})| u$ and $\gcd(n, \frac{q^{w^2}-1}{q-1})\nmid u$ if and only if $u=\gcd(n_{}, \frac{q^{w^2}-1}{q^w-1})u_1$, $1\le u_1\le \gcd(n, q^w-1)$, $\gcd(u_1,t)=1$ and $\frac{q^{w}-1}{q-1}\nmid u_1l_{w}$.
Hence by Lemma 2.4,
there exist $\frac{\varphi(t)}{wt}\cdot
(\gcd(n,q^w-1)-\gcd(n, q-1)) $  irreducible polynomials  of degree $wt$ in $\Bbb F_q[x]$.

  If  $x^t-\delta^{ul_{w^2}}\notin \Bbb F_{q^w}[x]$,  then by Lemma 2.4, the product of its $w^2$ conjugate polynomials is irreducible over $\Bbb F_q$.
 Hence there are  $\frac{\varphi(t)}{w^2t}\cdot (\gcd(n, q^{w^2}-1)-\gcd(n, q^w-1)$
   polynomials of degree $w^2t$ in $\Bbb F_q[x]$.

Suppose that  $t| m_{{w^2},2}$ and $t\nmid m_{{w^2},1}$. Similarly,  $x^t-\delta^{ul_{w^2}}\notin \Bbb F_{q}[x]$ and $x^t-\delta^{ul_{w^2}}\in \Bbb F_{q^w}[x]$
if and only if $u=\gcd(n_{}, \frac{q^{w^2}-1}{q^w-1})u_1$, $1\le u_1\le \gcd(n, q^w-1)$, $\gcd(u_1,t)=1$ and $\frac{q^{w}-1}{q-1}\nmid u_1l_{w}$. Hence
 there exist $\frac{\varphi(t)}{wt}\cdot
\gcd(n, q^w-1), \frac{\varphi(t)}{w^2t}\cdot (\gcd(n, q^{w^2}-1)-\gcd(n, q^w-1
))$   irreducible polynomials  of degrees $wt$  and   $w^2t$ in $\Bbb F_{q}[x]$, respectively.

    Suppose that  $t\nmid m_{w^2,2}$. Then
there is no  polynomial of degree $t$ and  $wt$ in $\Bbb F_q[x]$. So the number of irreducible polynomials of degree $w^2t$ in $\Bbb F_q[x]$ is
$\frac {\varphi(t)}{w^2t}\cdot \gcd(n,q^{w^2}-1).$

Hence we have the irreducible factorization of $x^n-1$ over $\Bbb F_q$.

Moreover, observe now that the function $\varphi(t)$ is a multiplicative function, then $\sum_{t\mid m_{w^2}}\frac {\varphi(t)}{t}$  is also
multiplicative and thus it is enough to calculate this sum for powers of primes. In this case
    we have
$\sum_{d|p^{k}}\frac{\varphi(d)}{d}=1+k(1-\frac 1 p).$ 
Finally, the number of irreducible factors of $x^n-1$ in $\Bbb F_q[x]$ is  
\begin{eqnarray*}&& \sum_{\mbox{\tiny$
\begin{array}{c}
t|m_{w^2,1}\\
\end{array}$
}}\frac{\varphi(t)}{w^2t}\cdot (\gcd(n, q^{w^2}-1)+(w-1)\gcd(n, q^{w}-1)+w(w-1)\gcd(n,q-1))\\
&+&\sum_{\mbox{\tiny$
\begin{array}{c}
t|m_{w^2,2}\\
t\nmid m_{w^2,1}\\
\end{array}$
}}\frac{\varphi(t)}{w^2t}\cdot (\gcd(n, q^{w^2}-1)+(w-1)\gcd(n, q^{w}-1))\\
&+&\sum_{\mbox{\tiny$
\begin{array}{c}
t \nmid m_{w^2,2}
\end{array}$
}}
\frac{\varphi(t)}{w^2t}\cdot\gcd(n, q^{w^2}-1)\\&=&\sum_{t|m_{w^2,1}}\frac{\varphi(t)}{t}\cdot \frac{w-1}{w}\cdot\gcd(n, q-1)+\sum_{t| m_{w^2,2}}\frac{\varphi(t)}{t}\cdot\frac{w-1}{w^2}\cdot\gcd(n, q^{w}-1)\\
&+& \sum_{t| m_{w^2}}\frac{\varphi(t)}{t}\cdot\frac{1}{w^2}\cdot\gcd(n, q^{w^2}-1)\\
&=&\prod_{\mbox{\tiny$
\begin{array}{c}
p| m_{w^2,1}\\
p~prime\\
\end{array}$
}}(1+v_p(m_{w^2,1})\frac{p-1}p)\cdot \frac{w-1}{w}\cdot\gcd(n, q-1) \end{eqnarray*}\begin{eqnarray*}
&+&\prod_{\mbox{\tiny$
\begin{array}{c}
p| m_{w^2,2}\\
p~prime\\
\end{array}$
}}(1+v_p(m_{w^2,2})\frac{p-1}p)\cdot \frac{w-1}{w^2}\cdot\gcd(n, q^w-1)\\
&+&\prod_{\mbox{\tiny$
\begin{array}{c}
p| m_{w^2}\\
p~prime\\
\end{array}$
}}(1+v_p(m_{w^2})\frac {p-1}p) \cdot\frac{1}{w^2}\cdot\gcd(n, q^{w^2}-1) .
\end{eqnarray*}

 {\bf Case 2.} $1=v_w(n)<v_w(q^{w^2}-1)$.  By Lemmas 2.5, 3.1, and 3.2,
$$\gcd(\frac{q^{w^2}-1}{q-1}, l_{w^2})=\frac{w\cdot\frac{q^{w^2}-1}{q-1}}{\gcd(n,\frac{q^{w^2}-1}{q-1})}\mbox{ and }\gcd(\frac{q^{w^2}-1}{q^w-1}, l_{w^2})=\frac{w\cdot\frac{q^{w^2}-1}{q^w-1}}{\gcd(n, \frac{q^{w^2}-1}{q^w-1})}.$$
Hence $x^t-\delta^{ul_{w^2}}\in \Bbb F_q[x]$ is irreducible   if and only if $\frac{\gcd(n, \frac{q^{w^2}-1}{q-1})}{w}\mid u$;
 $x^t-\delta^{ul_{w^2}}\notin \Bbb F_{q}[x]  $  and   $x^t-\delta^{ul_{w^2}}\in \Bbb F_{q^w}[x] $   if and only if $\frac {\gcd(n, \frac{q^{w^2}-1}{q^w-1})} {w}\mid u$ and $\frac{\gcd(n, \frac{q^{w^2}-1}{q-1})}{w}\nmid u$ .



{\bf Case 3.}   $2 \leqslant v_w(n)<v_w(q^{w^2}-1)$.
If  $v_w(q^{w^2}-1)-v_w(n)=1$,  then
$$\gcd(\frac{q^{w^2}-1}{q-1}, l_{w^2})=\frac{w\cdot\frac{q^{w^2}-1}{q-1}}{\gcd(n,\frac{q^{w^2}-1}{q-1})}\mbox{  and }\gcd(\frac{q^{w^2}-1}{q^w-1}, l_{w^2})=\frac{w\cdot\frac{q^{w^2}-1}{q^w-1}}{\gcd(n, \frac{q^{w^2}-1}{q^w-1})}.$$


If  $v_w(q^{w^2}-1)-v_w(n)\geqslant 2$,  then
$$\gcd(\frac{q^{w^2}-1}{q-1}, l_{w^2})=\frac{w^2\cdot\frac{q^{w^2}-1}{q-1}}{\gcd(n,\frac{q^{w^2}-1}{q-1})}\mbox{ and }\gcd(\frac{q^{w^2}-1}{q^w-1}, l_{w^2})=\frac{w\cdot\frac{q^{w^2}-1}{q^w-1}}{\gcd(n, \frac{q^{w^2}-1}{q^w-1})}.$$
Hence $x^t-\delta^{ul_{w^2}}\in \Bbb F_q[x]$ is irreducible   if and only if $\frac{\gcd(n, \frac{q^{w^2}-1}{q-1})}{w^2}\mid u$;  $x^t-\delta^{ul_{w^2}}\notin \Bbb F_{q}[x]  $ and  $x^t-\delta^{ul_{w^2}}\in \Bbb F_{q^w}[x]$   if and only if $\frac{\gcd(n, \frac{q^{w^2}-1}{q^w-1})}{w^2}\mid u$ and $\frac{\gcd(n, \frac{q^{w^2}-1}{q-1})}{w^2}\nmid u$.

Similarly, we can prove them.
\end{proof}

 \begin{exa}{\rm Suppose that $w=3, q=2$ and $n=73$. Then these positive integers satisfy the condition in Theorem $3.2$, i.e. $ord_{rad(n)}(q)=9$ and either $q\not\equiv3 \pmod 4$ or $8\nmid n$. By  calculate directly, $m_{w^2}=m_{w^2,1}=m_{w^2,2}=1$. Then the  irreducible factor number of $x^{73}-1$ over $\Bbb F_2$ is $\frac{73}9+\frac 2 9+\frac 2 3=9$ by Theorem 3.3, which is confirmed by Magma.  }
 \end{exa}

Second, we consider the case: $q\equiv 3\pmod 4$ and $8|n$.  Since $w$ is an odd prime,  $q^{w^2} \equiv3 \pmod 4$  if and only if $q \equiv3 \pmod 4$.

Let $r=\min\{v_2(n/2),v_2(q^{w^2}+1)\}=\min\{v_2(n/2),v_2(q+1)\}$, $m_{2{w^2}}=\frac{n}{\gcd(n, q^{2{w^2}}-1)}$, $l_{2{w^2}}=\frac{q^{2{w^2}}-1}{\gcd(n, q^{2{w^2}}-1)}$, $l_{w^2}=\frac{q^{w^2}-1}{\gcd(n, q^{w^2}-1)}$, $\lambda$ a  generator of $\Bbb F_{q^{2{w^2}}}^*$, and   $\delta=\lambda^{q^{w^2}+1}$.
Then the factorization of $x^{n} -1$ into irreducible factors in $\Bbb F_{q^{w^2}} [x ]$ is

\begin{equation}
\prod_{\mbox{\tiny$
\begin{array}{c}
t|m_{2{w^2}}\\
t~odd\\
\end{array}$
}} \prod_{\mbox{\tiny$
\begin{array}{c}
1\leqslant v\leqslant \gcd(n,q^{w^2}-1)\\
\gcd(v,t)=1\\
\end{array}$
}}(x^{t}-\delta^{vl_{{w^2}}})
\cdot\prod\limits_{t\mid m_{2{w^2}}} \prod\limits_{u\in \mathcal{R}_{t}} (x^{2t}-(\lambda^{ul_{2{w^2}}}+\lambda^{q^{w^2}ul_{2{w^2}}})x^{t}+\delta^{ul_{2{w^2}}})
\end{equation}
where
$$\mathcal{R}_t=\bigg\{u\in \Bbb N: \begin{array}{l}
1\le u\le \gcd(n,q^{2{w^2}}-1),\gcd(u,t)=1,
 \\ 2^{r}\nmid u,  u=min\{u, q^{w^2}u\}_{\gcd(n,q^{2{w^2}}-1)}\end{array}\bigg\}.$$

  Next we will explicitly give the irreducible factorization of $x^{n}-1$ over $\Bbb F_{q}$.

\begin{thm} {\rm Suppose that $w$ is an odd prime, $ord_{rad(n)}(q)=w^2$,  $q^{} \equiv 3\pmod 4$,  and $8| n$.
 Let $n=w^{v_w(n)}n_1n_2n_3$, $rad(n_1)|(q-1)$, $rad(n_2)|\frac{q^{w}-1}{q-1}$, and $rad(n_3)|\frac{q^{w^2}-1}{q^w-1}$. Let $m_{2w^2}=\frac{n}{\gcd(n,q^{2w^2}-1)}$, $m_{w^2}=\frac{n}{\gcd(n,q^{w^2}-1)}$,  $m_{{w^2},1}=\frac{n_1}{\gcd(n_1, q-1)}$, $m_{{w^2},2}=\frac{n_1n_2}{\gcd(n_1n_2, q^w-1)}$,
$l_{2w^2}=\frac{q^{2w^2}-1}{\gcd(n, q^{2w^2}-1)}$, $l_{w^2}=\frac{q^{w^2}-1}{\gcd(n,q^{w^2}-1)}$, $l_{w}=\frac{q^{w}-1}{\gcd(n,q^{w}-1)}$, $l_2=\frac{q^2-1}{\gcd(n,q^2-1)}$,  $l_{1}=\frac{q^{}-1}{\gcd(n,q^{}-1)}$, $r=min\{v_{2}(\frac n 2),v_{2}(q^{{w^2}}+1)\}$, $\lambda$ a  generator of $\Bbb F_{q^{2{w^2}}}^\ast$ satisfying $\delta=\lambda^{q^{{w^2}}+1}$, $\pi=\lambda^{\frac{q^{2{w^2}}-1}{q^w-1}}$, $\alpha=\lambda^{\frac{q^{2{w^2}}-1}{q^2-1}}$, and $\theta=\lambda^{\frac{q^{2w^2}-1}{q-1}}$. Then

$(1)$ The irreducible factorization of $x^{n}-1$ over $\Bbb F_{q}$ is

\begin{equation*}
\begin{split}
&\prod_{\mbox{\tiny$
\begin{array}{c}
 t|m_{w^2,1}\\
 t~ {odd}
\end{array}$
}}
 \prod_{\mbox{\tiny$
\begin{array}{c}
1\le v\le \gcd(n, q-1)\\
\gcd(v, t)=1
\end{array}$
}}(x^{t}-\theta^{vl_{1}})\\
&\cdot\prod_{\mbox{\tiny$
\begin{array}{c}
t|m_{w^2,2}\\  t~ {odd}
\end{array}$
}}\prod_{\mbox{\tiny$
\begin{array}{c}
v_1\in \mathcal{S}^{(1)}_{t}
\end{array}$
}}
\prod\limits_{k=0}^{w-1} (x^{t}-\pi^{q^kv_1l_{w}})
\cdot\prod_{\mbox{\tiny$
\begin{array}{c}
t|m_{2w^2}\\  t~ {odd}
\end{array}$
}}\prod_{\mbox{\tiny$
\begin{array}{c}
v_2\in \mathcal{S}^{(2)}_{t}
\end{array}$
}}
\prod\limits_{k=0}^{w^2-1} (x^{t}-\delta^{q^kv_2l_{w^2}})\\
& \cdot\prod_{\mbox{\tiny$
\begin{array}{c}
t| m_{w^2,1}\\u_1\in \mathcal{R}^{(1)}_t
\end{array}$
}}
(x^{2t}-(\alpha^{u_1l_{2}}+\alpha^{qu_1l_{2}})x^{t}+\theta^{u_1l_{2}}) \cdot\prod_{\mbox{\tiny$
\begin{array}{c}
t| m_{2w^2}\\
u_2\in\mathcal{ R}^{(2)}_{t}
\end{array}$
}}\prod\limits_{k=0}^{2w^2-1}(x^t-\lambda^{q^ku_2l_{2w^2}}),
\end{split}
\end{equation*}
where~ $$\mathcal{S}^{(1)}_{t}=\bigg\{v_1\in \Bbb N: \begin{array}{l}
1\le v_1\le \gcd(n,q^{w}-1),\gcd(v_1,t)=1,
 \\ \frac{q^{w}-1}{q-1}\nmid v_1l_{w}, v_1=\min\{v_1, qv_1, \cdots, q^{w-1}v_1\}_{\gcd(n, q^w-1)}\end{array}\bigg\},$$
 $$\mathcal{S}^{(2)}_{t}=\bigg\{v_2\in \Bbb N: \begin{array}{l}
1\le v_2\le \gcd(n,q^{w^2}-1),\gcd(v_2,t)=1,
 \\ \frac{q^{w^2}-1}{q^w-1}\nmid v_2l_{w^2}, v_2=\min\{v_2, qv_2, \cdots, q^{w^2-1}v_2\}_{\gcd(n, q^{w^2}-1)}\end{array}\bigg\},$$
$$\mathcal{R}^{(1)}_t=\bigg\{u_1\in \Bbb N: \begin{array}{l}
1\le u_1\le 2^{r}\gcd(n,q^{}-1),\gcd(u_1,t)=1,
 \\ 2^r\nmid u_1, u_1=\min\{u_1,q^{w^2}u_1\}_{\gcd(n, q^2-1)}\end{array}\bigg\},$$
and~$$ \mathcal{R}_t^{(2)}=\bigg\{u_2\in \Bbb N: \begin{array}{l}
1\le u_2\le \gcd(n,q^{2w^2}-1),\gcd(u_2,t)=1,   \frac{q^{2{w^2}}-1}{q^2-1}\nmid u_2l_{2{w^2}}
 \\  2^{r}\nmid u_2,  u_2=\min\{u_2, qu_2,\cdots,  q^{2w^2-1}u_2\}_{\gcd(n,q^{2w^2}-1)}\end{array}\bigg\}.$$


$(2)$
The number of irreducible factors of $x^n-1$ in $\Bbb F_q[x]$ is
\begin{eqnarray*}
&&\prod_{\mbox{\tiny$
\begin{array}{c}
p|m_{w^2,1}\\
p~odd~prime
\end{array}$
}}(1+v_p(m_{2w^2})\frac {p-1}p)\cdot \frac{A}{2w^2}\cdot\gcd(n, q-1)\\
&+&\prod_{\mbox{\tiny$
\begin{array}{c}
p|m_{w^2,2}\\
\end{array}$
}}
(1+v_p(m_{2w^2})\frac {p-1}p)\cdot \frac{w-1}{w^2}\cdot\gcd(n,q^{w}-1)\\&+&\prod_{\mbox{\tiny$
\begin{array}{c}
p|m_{2w^2}\\
p~odd~prime
\end{array}$
}}
(1+v_p(m_{2w^2})\frac {p-1}p)\cdot \frac{2^r+1+v_2(m_{2w^2})2^{r-1}}{2w^2}\cdot\gcd(n,q^{w^2}-1),
\end{eqnarray*}
where $A=2^{r-1}(2+v_2(m_{2w^2}))(w^2-1)+w^2-2w+1$.}
\end{thm}

\begin{proof}
Since $w$ is an odd integer and  $rad(n)|(q^{{w^2}}-1)$,  $$\gcd(n,q^{2w^2}-1)=\gcd(n/2, q^{{w^2}}+1)\gcd(n,q^{w^2}-1)=2^{r}\gcd(n,q^{w^2}-1),$$ where 
$r=\min\{v_2(n/2), v_2(q^{w^2}+1)\}=\min\{v_2(n/2),v_2(q+1)\}$.

Let $n=w^{v_w(n)}n_1n_2n_3$, $rad(n_1)|(q-1)$, $rad(n_2)|\frac{q^{w}-1}{q-1}$, and $rad(n_3)|\frac{q^{w^2}-1}{q^w-1}$. Let $m_{{w^2},1}=\frac{n_1}{\gcd(n_1, q-1)}$ and
$$m_{{w^2},2}=\frac{n_{1}n_2}{\gcd(n_1n_2,q^w-1)}=m_{{w^2},1}\cdot \frac{n_{2}}{\gcd(n_2,\frac{q^w-1}{q-1})}.$$
 Then $$\mbox{  }l_{2{w^2}}=\frac{q^{2{w^2}}-1}{\gcd(n,q^{2{w^2}}-1)}=\frac{q^{{w^2}}-1}{\gcd(n,q^{{w^2}}-1)}\cdot \frac{q^{{w^2}}+1}{\gcd(n/2,q^{{w^2}}+1)}=\frac{q^{{w^2}}+1}{2^{r}}l_{{w^2}} ,$$
  $$m_{2{w^2}}=\frac{n}{\gcd(n,q^{2{w^2}}-1)}=\frac{n}{\gcd(n/2,q^{{w^2}}+1)\gcd(n,q^{{w^2}}-1)}=\frac{m_{{w^2}}}{2^{r}},$$ $$ \mbox{ and }m_{w^2}=\frac{n}{\gcd(n,q^{w^2}-1)}=w^{v_w(n)-v_w(\gcd(n, q^{w^2}-1))}\cdot m_{w^2,2}\cdot  \frac{n_{3}}{\gcd(n_3,\frac{q^{w^2}-1}{q^w-1})}.$$


Firstly, we investigate the product in (3.3): 

$$\prod_{\mbox{\tiny$
\begin{array}{c}
t| m_{2{w^2}}\\
t~odd\\
\end{array}$
}} \prod_{\mbox{\tiny$
\begin{array}{c}
1\leqslant v\leqslant \gcd(n,q^{w^2}-1)\\
\gcd(v,t)=1\\
\end{array}$
}}(x^{t}-\delta^{vl_{{w^2}}}).$$
Note that $t$ is always odd.
Similarly,   we can obtain all irreducible polynomials of degrees $t$, $wt$ and $w^2t$ as follows.

Suppose that $t| m_{{w^2},1}$. Then the numbers of irreducible polynomials of degree $t$, $wt$ and $w^2t$ in $\Bbb F_q[x]$ are
 $$\frac{\varphi(t)}t\cdot
\gcd(n, q-1),\ \frac{\varphi(t)}{wt}\cdot
(\gcd(n,q^w-1)-\gcd(n, q-1)), $$ and   $$\frac{\varphi(t)}{w^2t}\cdot (\gcd(n, q^{w^2}-1)-\gcd(n, q^w-1)).$$

Suppose that  $t| m_{{w^2},2}$ and $t\nmid m_{{w^2},1}$. Then
there is no  polynomial of degree $t$ in $\Bbb F_q[x]$. Hence
 there exist $$\frac{\varphi(t)}{wt}\cdot
\gcd(n, q^w-1),\ \frac{\varphi(t)}{w^2t}\cdot (\gcd(n, q^{w^2}-1)-\gcd(n, q^w-1
))$$   irreducible polynomials  of degree $wt$  and      of degree $w^2t$ in $\Bbb F_{q}[x]$, respectively.

Suppose that  $t\nmid m_{w^2,2}$. Then
there is no  polynomial of degree $t$ and  $wt$ in $\Bbb F_q[x]$. So the number of irreducible polynomials of degree $w^2t$ in $\Bbb F_q[x]$ is
$\frac {\varphi(t)}{w^2t}\cdot \gcd(n,q^{w^2}-1).$

Secondly, we investigate the product in (3.3):

\begin{equation*}\prod\limits_{t\mid m_{2{w^2}}} \prod\limits_{u\in \mathcal{R}_{t}} (x^{2t}-(\lambda^{ul_{2{w^2}}}+\lambda^{q^{w^2}ul_{2{w^2}}})x^{t}+\delta^{ul_{2{w^2}}}). \end{equation*}



Next we proceed the proof by three cases.


 {\bf Case 1.}  $v_w(n)=0$ or $v_w(n)\geqslant v_w(q^{w^2}-1)$.
By Lemmas 2.5, 3.1 and 3.2,
\begin{eqnarray*}
\gcd(\frac{q^{2w^2}-1}{q-1}, l_{2w^2})=\frac{\frac{q^{2w^2}-1}{q-1}}{2^r\gcd(n,\frac{q^{w^2}-1}{q-1})}\mbox{ and } \gcd(\frac{q^{2{w^2}}-1}{q^2-1}, l_{2{w^2}})=\frac{\frac{q^{2{w^2}}-1}{q^2-1}}{\gcd(n, \frac{q^{w^2}-1}{q-1})}.
\end{eqnarray*}
Note that the irreducible polynomial $x^{2t}-(\lambda^{ul_{2{w^2}}}+\lambda^{q^{{w^2}}ul_{2{w^2}}})x^{t}+\delta^{ul_{2{w^2}}}=(x^t-\lambda^{ul_{2{w^2}}})(x^t-\lambda^{q^{w^2}ul_{2{w^2}}})\in \Bbb F_{q^{w^2}}[x]$ and $x^t-\lambda^{ul_{2{w^2}}}\notin  \Bbb F_{q} [x]$. Hence  $x^{2t}-(\lambda^{ul_{2{w^2}}}+\lambda^{q^{{w^2}}ul_{2{w^2}}})x^{t}+\delta^{ul_{2{w^2}}}\in \Bbb F_{q}[x]$ if and only if $x^t-\lambda^{ul_{2{w^2}}}\in \Bbb F_{q^2}[x]$ if and only if $\frac{q^{2{w^2}}-1}{q^2-1}|ul_{2{w^2}}$ and $\frac{q^{2{w^2}}-1}{q-1}\nmid ul_{2{w^2}}$  if and only if $u=\gcd(n_{}, \frac{q^{w^2}-1}{q-1})u_1$ for $1\le u_1\le 2^r\gcd(n, q-1)$, $\gcd(u_1,t)=1$ and $2^r\nmid u_1$, i.e.  $x^{2t}-(\lambda^{ul_{2{w^2}}}+\lambda^{q^{{w^2}}ul_{2{w^2}}})x^{t}+\delta^{ul_{2{w^2}}}=x^{2t}-(\alpha^{u_1l_2}+\alpha^{qu_1l_2})x^t+\theta^{u_1l_2}\in \Bbb F_q[x]$ is irreducible.

In Lemma 2.3, the number of irreducible polynomials of degree $2t$  depends on the parity of $t$. In the following, we divide the proof into two cases.

Suppose that $t|m_{w^2, 1}$.  If $t$ is odd,  there exist $\frac{\varphi(t)}{2t}\cdot(2^{r_{}}-1)\cdot\gcd(n,q-1)$ irreducible polynomials of degree $2t$ in $\Bbb F_{q^{w^2}}[x]$ that are in $\Bbb F_q[x]$.
  On the other hand, each irreducible polynomial  $x^{2t}-(\lambda^{ul_{2{w^2}}}+\lambda^{q^{{w^2}}ul_{2{w^2}}})x^{t}+\delta^{ul_{2{w^2}}}\in \Bbb F_{q^{w^2}}[x]$ with $\frac{q^{2{w^2}}-1}{q^2-1}\nmid ul_{2{w^2}}$,  a product of its $w^2$  conjugate irreducible polynomials  is irreducible in $\Bbb F_q[x]$, i.e. $\prod_{k=0}^{2w^2-1}(x^t-\lambda^{q^kul_{2{w^2}}})\in \Bbb F_q[x]$ is irreducible. From Lemma 2.3, the number of irreducible polynomials of degree $2t$ over $\Bbb F_{q^{w^2}}$ in (3.3) is $\frac{\varphi(t)}{2t}\cdot(2^{r_{}}-1)\cdot\gcd(n,q^{{w^2}}-1)$. Hence
   the number of irreducible polynomials of degree $2{w^2}t$ over  $\Bbb F_{q}$ is $$\frac{\varphi(t)}{2{w^2}t}\cdot(2^{r_{}}-1)\cdot(\gcd(n,q^{{w^2}}-1)-\gcd(n,q-1)).$$
If $t$ is even,  then $u_1$ must be an odd integer and the condition $2^r\nmid u_1$ automatically holds. From Lemma 2.3, the number of irreducible polynomials of degree $2t$ over $\Bbb F_{q^{w^2}}$ in (3.3) is $\frac{\varphi(t)}{2t}\cdot2^{r}\cdot\gcd(n,q^{{w^2}}-1).$ Hence, the numbers of irreducible polynomials of degree $2t$ and $2{w^2}t$ over  $\Bbb F_{q}$ are $$\frac{\varphi(t)}{2t}\cdot2^{r_{}}\cdot\gcd(n,q-1)\mbox { and }\frac{\varphi(t)}{2{w^2}t}\cdot2^{r_{}}\cdot(\gcd(n,q^{{w^2}}-1)-\gcd(n,q-1)).$$

Suppose that  $t\nmid m_{{w^2},1}$. Then there is no polynomial of degree $2t$ in $\Bbb F_q[x]$.   If $t$ is odd, then the number of irreducible polynomials of degree $2{w^2}t$ in $\Bbb F_q[x]$ is
$\frac{\varphi(t)}{2{w^2}t}\cdot(2^{r_{}}-1)\cdot\gcd(n,q^{{w^2}}-1).$
If $t$ is even, then the number of irreducible polynomials of degree $2w^2t$ in $\Bbb F_q[x]$ is
$\frac{\varphi(t)}{2{w^2}t}\cdot2^{r}\cdot\gcd(n,q^{{w^2}}-1).$

 Hence we have the irreducible factorization of $x^n-1$ over $\Bbb F_q$. Finally,  the number of irreducible factors of $x^n-1$ in $\Bbb F_q[x]$ is
\begin{eqnarray*}&&\sum_{\mbox{\tiny$
\begin{array}{c}
t|m_{{w^2},1}\\
t~odd
\end{array}$
}}
\frac{\varphi(t)}{2{w^2}t}\cdot((2^r+1)\gcd(n, q^{w^2}-1)+2(w-1)\gcd(n,q^w-1)\\
&+&((2^r-1)(w^2-1)+2w(w-1))\gcd(n, q-1))\\
&+&\sum_{\mbox{\tiny$
\begin{array}{c}
t|m_{{w^2},1}\\
t~even
\end{array}$
}}
\frac{\varphi(t)}{2{w^2}t}\cdot2^r\cdot(\gcd(n, q^{w^2}-1)+(w^2-1)\gcd(n, q-1))\\
&+&\sum_{\mbox{\tiny$
\begin{array}{c}
t|m_{{w^2},2}\\
t\nmid m_{{w^2},1}\\
t~odd
\end{array}$
}}
\frac{\varphi(t)}{2{w^2}t}\cdot((2^r+1)\gcd(n, q^{w^2}-1)+2(w-1)\gcd(n,q^w-1))\\&+&\sum_{\mbox{\tiny$
\begin{array}{c}
t|m_{{w^2},2}\\t\nmid m_{{w^2},1}\\
t~even
\end{array}$
}}
\frac{\varphi(t)}{2{w^2}t}\cdot2^r\cdot\gcd(n, q^{w^2}-1)
+\sum_{\mbox{\tiny$
\begin{array}{c}
t\nmid m_{w^2,2}\\
t~odd\\
\end{array}$
}}
\frac{\varphi(t)}{2w^2t}\cdot(2^r+1)\gcd(n, q^{w^2}-1)\\
&+&\sum_{\mbox{\tiny$
\begin{array}{c}
t\nmid m_{w^2,2}\\
t~even\\
\end{array}$
}}
\frac{\varphi(t)}{2w^2t}\cdot2^r\cdot\gcd(n, q^{w^2}-1)\\
&=&\prod_{\mbox{\tiny$
\begin{array}{c}
p|m_{w^2,1}\\
p~odd~prime
\end{array}$
}}(1+v_p(m_{2w^2})\frac {p-1}p)\cdot \frac{A}{2w^2}\cdot\gcd(n, q-1)\\
&+&\prod_{\mbox{\tiny$
\begin{array}{c}
p|m_{w^2,2}\\
\end{array}$
}}
(1+v_p(m_{2w^2})\frac {p-1}p)\cdot \frac{w-1}{w^2}\cdot\gcd(n,q^{w}-1)\\
&+&\prod_{\mbox{\tiny$
\begin{array}{c}
p|m_{2w^2}\\
p~odd~prime
\end{array}$
}}
(1+v_p(m_{2w^2})\frac {p-1}p)\cdot \frac{2^r+1+v_2(m_{2w^2})2^{r-1}}{2w^2}\cdot\gcd(n,q^{w^2}-1),
\end{eqnarray*}
where $A=2^{r-1}(2+v_2(m_{2w^2}))(w^2-1)+w^2-2w+1$.

 {\bf Case 2.}  $1=v_w(n)<v_w(q^{w^2}-1)$.
By Lemmas 2.5, 3.1, and 3.2,
$$\gcd(\frac{q^{2w^2}-1}{q-1}, l_{2w^2})=\frac{w\cdot\frac{q^{2w^2}-1}{q-1}}{2^r\gcd(n,\frac{q^{w^2}-1}{q-1})}\mbox{ and } \gcd(\frac{q^{2{w^2}}-1}{q^2-1}, l_{2{w^2}})=\frac{w\cdot\frac{q^{2{w^2}}-1}{q^2-1}}{\gcd(n, \frac{q^{w^2}-1}{q-1})}. $$



{\bf Case 3.}  $2 \leqslant v_w(n)<v_w(q^{w^2}-1)$.
If  $v_w(q^{w^2}-1)-v_w(n)=1$,  then
$$\gcd(\frac{q^{2w^2}-1}{q-1}, l_{2w^2})=\frac{w\cdot\frac{q^{2w^2}-1}{q-1}}{2^r\gcd(n,\frac{q^{w^2}-1}{q-1})}\mbox{ and }
\gcd(\frac{q^{2{w^2}}-1}{q^2-1}, l_{2{w^2}})=\frac{w\cdot\frac{q^{2{w^2}}-1}{q^2-1}}{\gcd(n, \frac{q^{w^2}-1}{q-1})}.$$

If  $v_w(q^{w^2}-1)-v_w(n)\geqslant 2$,  then
$$\gcd(\frac{q^{2w^2}-1}{q-1}, l_{2w^2})=\frac{w^2\cdot\frac{q^{2w^2}-1}{q-1}}{2^r\gcd(n,\frac{q^{w^2}-1}{q-1})}\mbox{ and } \gcd(\frac{q^{2{w^2}}-1}{q^2-1}, l_{2{w^2}})=\frac{w^2\cdot\frac{q^{2{w^2}}-1}{q^2-1}}{\gcd(n, \frac{q^{w^2}-1}{q-1})}.$$

Similarly, we can prove them.
 \end{proof}

\begin{exa} {\rm Suppose that $q=w=3$ and $n=6056$.  It is easy to check that these positive integers satisfying the condition in Theorem $3.4$, i.e. $ord_{rad(n)}(q)=9$ and  $q\equiv3 \pmod 4$ and $8| n$.
By  calculate directly,  $r=2$ and $m_{2w^2}=1$, $m_{w^2,1}=m_{w^2,2}=4$.  Hence the number of  irreducible factors  of $x^{n}-1$ over $\Bbb F_{3}$ is
$\frac{10\times 757}{2\times9}+\frac{4}{9}+\frac{32+12}{9}-\frac{8}{9}=425$ by Theorem 3.5. The result  is confirmed by Magma.  }


\end{exa}

\subsection{$w=2$.}$~$

 In the subsection, we always assume  that  $\ord_{rad(n)}(q)=4$.

 Similar to Lemmas 3.1 and 3.2,  the following lemma  is easy to verify.
\begin{lem}




{\rm
$(1)$ For $b\in \{0,1,2\}$, if $\gcd(n,q-1, \frac{q^{4}-1}{q-1})=2^{b} $, then
$$\gcd(n, q^{4}-1)=\left\{\begin{array}{ll}
\gcd(n, q^{}-1)\gcd(n, \frac{q^{4}-1}{q-1}), &\mbox{ if $v_2(n)\geqslant v_2(q^4-1)$ or $b=0$,}\\
\gcd(n, q-1)\gcd(n/2, \frac{q^{4}-1}{q-1}), &\mbox{ if $v_2(n)< v_2(q^4-1)$ and $b=1$,}\\
\gcd(n/4, q-1)\gcd(n, \frac{q^{4}-1}{q-1}), &\mbox{ if $v_2(n)< v_2(q^4-1)$ and $b=2$.}\\
\end{array}\right.$$

$(2)$   For $b\in \{0,1,2\}$, if $\gcd(n,q-1, \frac{q^{4}-1}{q-1})=2^{b} $, then
$$\gcd(n, q^{4}-1)=\left\{\begin{array}{ll}
\gcd(n, q^{2}-1)\gcd(n, \frac{q^{4}-1}{q^2-1}), &\mbox{ if $v_2(n)\geqslant v_2(q^4-1)$ or $b=0$,}\\
\gcd(n/2, q^2-1)\gcd(n, \frac{q^{4}-1}{q^2-1}), &\mbox{ if $v_2(n)< v_2(q^4-1)$ and $b\neq0$.}\\
\end{array}\right.$$}

\end{lem}


By Lemmas 2.5, 3.8 and the proof of Theorem 3.3,  the following result holds.




\begin{thm}{\rm Suppose that $\ord_{rad(n)}(q)=4$.
 Let $n=2^{v_2(n)}n_1n_2n_3$, $rad(n_1)|(q-1)$, $rad(n_2)|(q+1)$, and $rad(n_3)|\frac{q^{4}-1}{q^2-1}$. Let $m_{4}=\frac{n}{\gcd(n,q^{4}-1)}$, $m_{{4},1}=\frac{n_1}{\gcd(n_1, q-1)}$, $m_{{4},2}=\frac{n_1n_2}{\gcd(n_1n_2, q^2-1)}$, $l_{4}=\frac{q^{4}-1}{\gcd(n,q^{4}-1)}$, $l_{2}=\frac{q^{2}-1}{\gcd(n,q^{2}-1)}$, $l_{1}=\frac{q-1}{\gcd(n, q-1)}$,  $\mu$ a generator of $\Bbb F^{\ast}_{q^{4}}$ satisfying $\alpha=\mu^{q^2+1}$ and $\theta= \mu^{\frac{q^4-1}{q-1}}$.
Then

$(1)$ The irreducible factorization of  $x^{n}-1$  over $\Bbb F_{q}$  is
\begin{eqnarray*}&&\prod_{\mbox{\tiny$
\begin{array}{c}
 t|m_{4,1}\\
\end{array}$
}}
 \prod_{\mbox{\tiny$
\begin{array}{c}
1\le u\le \gcd(n, q-1)\\
\gcd(u, t)=1
\end{array}$
}}(x^{t}-\theta^{ul_{1}})\\
&\cdot&\prod_{\mbox{\tiny$
\begin{array}{c}
t|m_{4,2}\\
\end{array}$
}}\prod_{\mbox{\tiny$
\begin{array}{c}
u_1\in \mathcal{S}^{(1)}_{t}
\end{array}$
}}
 (x^{2t}-(\alpha^{u_1l_{2}}+\alpha^{qu_1l_{2}})x^t+\theta^{u_1l_2})
\cdot\prod_{\mbox{\tiny$
\begin{array}{c}
t|m_{4}\\
\end{array}$
}}\prod_{\mbox{\tiny$
\begin{array}{c}
u_2\in \mathcal{S}^{(2)}_{t}
\end{array}$
}}
\prod\limits_{k=0}^{3} (x^{t}-\mu^{q^ku_2l_{4}})
\end{eqnarray*}
 where~$$\mathcal{S}^{(1)}_{t}=\bigg\{u_1\in \Bbb N: \begin{array}{l}
1\le u_1\le \gcd(n,q^{2}-1),\gcd(u_1,t)=1,
 \\ (q+1)\nmid u_1l_{2}, u_1=\min\{u_1, qu_1\}_{\gcd(n, q^2-1)}\end{array}\bigg\},$$
 $$\mathcal{S}^{(2)}_{t}=\bigg\{u_2\in \Bbb N: \begin{array}{l}
1\le u_2\le \gcd(n,q^{4}-1),\gcd(u_2,t)=1,
 \\ \frac{q^{4}-1}{q^2-1}\nmid u_2l_{4}, u_2=\min\{u_2, qu_2, q^2u_2, q^{3}u_2\}_{\gcd(n, q^{4}-1)}\end{array}\bigg\}.$$

$(2)$
The number of irreducible factors of $x^n-1$ in $\Bbb F_q[x]$ is
\begin{eqnarray*}&&\prod_{\mbox{\tiny$
\begin{array}{c}
p| m_{4,1}\\
p~prime\\
\end{array}$
}}(1+v_p(m_{4,1})\frac{p-1}p)\cdot \frac{\gcd(n, q-1) }{2} +\prod_{\mbox{\tiny$
\begin{array}{c}
p| m_{4,2}\\
p~prime\\
\end{array}$
}}(1+v_p(m_{4,2})\frac{p-1}p)\cdot \frac{\gcd(n, q^2-1) }{4}\\ &+&\prod_{\mbox{\tiny$
\begin{array}{c}
p| m_{4}\\
p~prime\\
\end{array}$
}}(1+v_p(m_{4})\frac {p-1}p) \cdot\frac{\gcd(n, q^{4}-1)}{4}.
\end{eqnarray*}   }

\end{thm}

\begin{exa}{\rm  Let $w=2,q=3$, and $n=40$. It is easy to check that $ord_{rad(n)}(q)=4$. By Theorem 3.8, the irreducible factor number of $x^{40}-1$ over $\Bbb F_3$ is $10+2+1=13$, which is confirmed by Magma.  }

\end{exa}
\section{Case: $w_1\neq w_2$}

In this section, we always assume that $\ord_{rad(n)}(q)=w_1w_2$, where $w_1,w_2$ are two distinct  primes.


Suppose that   $x^n-1$ has an irreducible factorization over $\Bbb F_{q^{w_1w_2}}$:
\begin{equation}x^{n}-1=\prod\limits_{d|n}\Phi_{d}(x)=\prod\limits_{d|n}f_{d,1}(x)f_{d,2}(x)\cdots f_{d,s_{d}}(x).\end{equation} For each positive  divisor $d$ of $n$,  the order of $q$ modulo $rad(d)$ is $1$,  $w_1$, $w_2$ or $w_1w_2$ by $rad(d)| rad(n)$ and $w_1,w_2$ are two distinct  primes. We define three sets  $$D_1=\{d\mid n: 1\le  d\le n, rad(d)|(q-1)\},D_{w_1}=\{d\mid n: 1\le  d\le n, d\notin D_1, rad(d)|(q^{w_1}-1)\},$$  $$D_{w_2}=\{d\mid n: 1\le  d\le n, d\notin D_1, rad(d)|(q^{w_2}-1)\}.$$ Hence those irreducible polynomials in (4.1) can be divided into four parts:
$$X_{(0)}=\{f_{d,j}(x):d\in D_1, 1\leqslant j\leqslant s_{d}\},X_{(1)}=\{f_{d,j}(x):d\in D_{w_1}, 1\leqslant j\leqslant s_{d}\},$$
$$X_{(2)}=\{f_{d,j}(x):d\in D_{w_2}, 1\leqslant j\leqslant s_{d}\},\mbox{  }X_{(3)}=\{f_{d,j}(x):d\notin D_1\cup D_{w_1}\cup D_{w_2}, 1\leqslant j\leqslant s_{d}\}.$$
There are equivalence relations $\sim_1$,  $\sim_2$ and   $\sim_3$ on $X_{(1)}$, $ X_{(2)}$ and $ X_{(3)}$ as follows:   $$f_{d,i}(x)\sim_1 f_{d,j}(x)\mbox{ {if}~ {and} ~only~ if~} f_{d,j}(x)= f_{d,i}^{\sigma^{k}}(x) \mbox{ for ~some~} 0 \leqslant k\leqslant w_1 - 1,$$ $$f_{d,i}(x)\sim_2 f_{d,j}(x)\mbox{ {if}~ {and} ~only~ if~} f_{d,j}(x)= f_{d,i}^{\sigma^{k}}(x)  \mbox{ for ~some~} 0 \leqslant k\leqslant w_2 - 1,$$  $$f_{d,i}(x)\sim_3 f_{d,j}(x)\mbox{ {if}~ {and} ~only~ if~} f_{d,j}(x)= f_{d,i}^{\sigma^{k}}(x)  \mbox{ for ~some~}0 \leqslant k\leqslant w_1w_2 - 1. $$
By  Lemma 2.4,  the irreducible factorization of $x^{n}-1$ over $\Bbb F_{q}$ is
$$\prod\limits_{f(x)\in X_{(0)} }f(x) \prod\limits_{g(x)\in X^{\ast}_{(1)} } \prod\limits_{k=0 }^{w_1-1}g^{\sigma^k}(x)\prod\limits_{h(x)\in X^{\ast}_{(2)} } \prod\limits_{k=0 }^{w_2-1}h^{\sigma^k}(x)\prod\limits_{i(x)\in X^{\ast}_{(3)} } \prod\limits_{k=0 }^{w_1w_2-1}i^{\sigma^k}(x),$$
 where $X^{\ast}_{(1)}$, $X^{\ast}_{(2)}$  and $X^{\ast}_{(3)}$ are complete systems  of equivalence class representatives of $X_{(1)}$, $X_{(2)}$ and $X_{(3)}$ relative to $\sim_1$,  $\sim_2$ and  $\sim_3$, respectively.

  In the following, we consider the irreducible factorization of $x^n-1$ over $\Bbb F_q$ by two cases: (1) $w_1w_2$ is  odd, (2) $w_1w_2$ is  even.
    \subsection{$w_1w_2$ is an odd number} $~$

    In the subsection, we always assume  that $w_1,w_2$ are two odd distinct primes, and $\ord_{rad(n)}(q)=w_1w_2$.

\begin{lem}{\rm Let $w_1,w_2$ be two odd distinct primes. Then



$(1)$ If either $v_{w_1}(n)=0 $ or $v_{w_1}(n)\geqslant v_{w_1}( q^{w_1w_2}-1)$ and either $v_{w_2}(n)=0 $ or $v_{w_2}(n)\geqslant v_{w_i}( q^{w_1w_2}-1)$, then
\begin{eqnarray*}
\gcd(n, q^{w_1w_2}-1)&=&\gcd(n, q^{}-1)\gcd(n, \frac{q^{w_1w_2}-1}{q-1})
=\gcd(n, q^{w_1}-1)\gcd(n, \frac{q^{w_1w_2}-1}{q^{w_1}-1})\\
&=&\gcd(n, q^{w_2}-1)\gcd(n, \frac{q^{w_1w_2}-1}{q^{w_2}-1}).
\end{eqnarray*}

$(2)$ If either $v_{w_1}(n)=0 $ or $v_{w_1}(n)\geqslant v_{w_1}( q^{w_1w_2}-1)$ and  $1\leqslant v_{w_2}(n)<v_{w_2}( q^{w_1w_2}-1)$, then
\begin{eqnarray*}
\gcd(n, q^{w_1w_2}-1)&=&\gcd(n/w_2, q^{}-1)\gcd(n, \frac{q^{w_1w_2}-1}{q-1})=\gcd(n/w_2, q^{w_1}-1)\gcd(n, \frac{q^{w_1w_2}-1}{q^{w_1}-1})\\
&=&\gcd(n, q^{w_2}-1)\gcd(n, \frac{q^{w_1w_2}-1}{q^{w_2}-1}).
\end{eqnarray*}

$(3)$ If either $v_{w_2}(n)=0 $ or $v_{w_2}(n)\geqslant v_{w_2}( q^{w_1w_2}-1)$ and $1\leqslant v_{w_1}(n)<v_{w_1}( q^{w_1w_2}-1)$, then
\begin{eqnarray*}
\gcd(n, q^{w_1w_2}-1)&=&\gcd(n/w_1, q^{}-1)\gcd(n, \frac{q^{w_1w_2}-1}{q-1})=\gcd(n, q^{w_1}-1)\gcd(n, \frac{q^{w_1w_2}-1}{q^{w_1}-1})\\
&=&\gcd(n/w_1, q^{w_2}-1)\gcd(n, \frac{q^{w_1w_2}-1}{q^{w_2}-1}).
\end{eqnarray*}

$(4)$ If $1\leqslant v_{w_1}(n)<v_{w_1}( q^{w_1w_2}-1)$  and $1\leqslant v_{w_2}(n)<v_{w_2}( q^{w_1w_2}-1)$, then
\begin{eqnarray*}
\gcd(n, q^{w_1w_2}-1)&=&\gcd(n/w_1w_2, q^{}-1)\gcd(n, \frac{q^{w_1w_2}-1}{q-1})=\gcd(n/w_2, q^{w_1}-1)\gcd(n, \frac{q^{w_1w_2}-1}{q^{w_1}-1})\\
&=&\gcd(n/w_1, q^{w_2}-1)\gcd(n, \frac{q^{w_1w_2}-1}{q^{w_2}-1}).
\end{eqnarray*}


}

\end{lem}



First, we consider the case:  either $q\not\equiv3\pmod 4$ or $8\nmid n$.
 Let $m_{w_1w_2}=\frac{n}{\gcd(n,q^{w_1w_2}-1)}$, $l_{w_1w_2}=\frac{q^{w_1w_2}-1}{\gcd(n,q^{w_1w_2}-1)}$, and  $\delta$ a  generator of $\Bbb F_{q^{w_1w_2}}^*$. By Lemma 2.2, there is an irreducible factorization of $x^n-1$ over $\Bbb F_{q^{{w_1w_2}}}$:
 \begin{equation}
 x^{n} -1=\prod\limits_{t|m_{w_1w_2}} \prod_{\mbox{\tiny$
\begin{array}{c}
1\leqslant u\leqslant \gcd(n,q^{{w_1w_2}}-1)\\
\gcd(u,t)=1\\
\end{array}$
}}(x^{t}-\delta^{ul_{w_1w_2}}).
\end{equation}

  Next we will explicitly give the irreducible factorization of $x^{n}-1$ over $\Bbb F_{q}$.

\begin{thm}{\rm  Suppose that $w_1,w_2$ are two odd distinct primes, $\ord_{rad(n)}(q)=w_1w_2$, and either $q\not\equiv3 \pmod 4$ or $8\nmid n$.
Let $n=w_1^{v_{w_1}(n)}w_2^{v_{w_2}(n)}n_0n_1n_2n_3$, where
$rad(n_0)|(q-1)$,  $rad(n_1)|\frac{q^{w_1}-1}{q-1}$, $rad(n_2)|\frac{q^{w_2}-1}{q-1}$, and $\gcd(n_3,(q^{w_1}-1)(q^{w_2}-1))=1$. Let
   $m_{{w_1w_2},0}=\frac{n_0}{\gcd(n_0, q-1)}$, $m_{{w_1w_2},1}=\frac{w_1^{v_{w_1}(n)}n_0n_1}{\gcd(w_1^{v_{w_1}(n)}n_0n_1, q^{w_1}-1)}$,  $m_{{w_1w_2},2}=\frac{w_2^{v_{w_2}(n)}n_0n_2}{\gcd(w_2^{v_{w_2}(n)}n_0n_2, q^{w_2}-1)}$, $m_{w_1w_2}=\frac{n}{\gcd(n,q^{w_1w_2}-1)}$, $l_{w_1w_2}=\frac{q^{w_1w_2}-1}{\gcd(n,q^{w_1w_2}-1)}$, $l_{w_1}=\frac{q^{w_1}-1}{\gcd(n,q^{w_1}-1)}$, $l_{w_2}=\frac{q^{w_2}-1}{\gcd(n,q^{w_2}-1)}$, $l_{1}=\frac{q-1}{\gcd(n, q-1)}$,  $\delta$ a  generator of $\Bbb F^{\ast}_{q^{w_1w_2}}$, $\pi_1$ a  generator of $\Bbb F^{\ast}_{q^{w_1}}$ satisfying  $\pi_1=\delta^{\frac{q^{w_1w_2}-1}{q^{w_1}-1}}$, and $\pi_2$ a  generator of $\Bbb F^{\ast}_{q^{w_2}}$ satisfying  $\pi_2=\delta^{\frac{q^{w_1w_2}-1}{q^{w_2}-1}}$. Then

$(1)$ The irreducible factorization of  $x^{n}-1$  over $\Bbb F_{q}$  is
\begin{eqnarray*}&&\prod_{\mbox{\tiny$
\begin{array}{c}
 t|m_{w_1w_2,0}\\
\end{array}$
}}
 \prod_{\mbox{\tiny$
\begin{array}{c}
1\le u\le \gcd(n, q-1)\\
\gcd(u, t)=1
\end{array}$
}}(x^{t}-\theta^{ul_{1}})\cdot\prod_{\mbox{\tiny$
\begin{array}{c}
t|m_{w_1w_2,1}\\
\end{array}$
}}\prod_{\mbox{\tiny$
\begin{array}{c}
u_1\in \mathcal{S}^{(1)}_{t}
\end{array}$
}}
\prod\limits_{k=0}^{w_1-1}(x^{t}-\pi_1^{q^ku_1l_{w_1}})\\
&\cdot&\prod_{\mbox{\tiny$
\begin{array}{c}
t|m_{w_1w_2,2}\\
\end{array}$
}}\prod_{\mbox{\tiny$
\begin{array}{c}
u_2\in \mathcal{S}^{(2)}_{t}
\end{array}$
}}
\prod\limits_{k=0}^{w_2-1}(x^{t}-\pi_2^{q^ku_2l_{w_2}})\cdot\prod_{\mbox{\tiny$
\begin{array}{c}
t|m_{w_1w_2}\\
\end{array}$
}}\prod_{\mbox{\tiny$
\begin{array}{c}
u_3\in \mathcal{S}^{(3)}_{t}
\end{array}$
}}
\prod\limits_{k=0}^{w_1w_2-1} (x^{t}-\delta^{q^ku_3l_{w_1w_2}}),
\end{eqnarray*}
 where~$$\mathcal{S}^{(1)}_{t}=\bigg\{u_1\in \Bbb N: \begin{array}{l}
1\le u_1\le \gcd(n,q^{w_1}-1),\gcd(u_1,t)=1,
 \\ \frac{q^{w_1}-1}{q-1}\nmid u_1l_{w_1}, u_1=\min\{u_1, qu_1, \cdots, q^{w_1-1}u_1\}_{\gcd(n, q^{w_1}-1)}\end{array}\bigg\},$$
 $$\mathcal{S}^{(2)}_{t}=\bigg\{u_2\in \Bbb N: \begin{array}{l}
1\le u_2\le \gcd(n,q^{w_2}-1),\gcd(u_2,t)=1,
 \\ \frac{q^{w_2}-1}{q-1}\nmid u_2l_{w_2}, u_2=\min\{u_2, qu_2, \cdots, q^{w_2-1}u_2\}_{\gcd(n, q^{w_2}-1)}\end{array}\bigg\},$$
  $$\mathcal{S}^{(3)}_{t}=\bigg\{u_3\in \Bbb N: \begin{array}{l}
1\le u_3\le \gcd(n,q^{w_1w_2}-1),\gcd(u_3,t)=1,
 \\ \frac{q^{w_1w_2}-1}{q^{w_1}-1}\nmid u_3l_{w_1w_2}, \frac{q^{w_1w_2}-1}{q^{w_2}-1}\nmid u_3l_{w_1w_2},\\u_3=\min\{u_3, qu_3, \cdots, q^{w_1w_2-1}u_3\}_{\gcd(n, q^{w_1w_2}-1)}\end{array}\bigg\}.$$
$(2)$ The number of irreducible factors is
\begin{eqnarray*}&&\prod_{\mbox{\tiny$
\begin{array}{c}
p| m_{w_1w_2,0}\\
p~prime\\
\end{array}$
}}(1+v_p(m_{w_1w_2,0})\frac{p-1}p)\cdot \frac{(w_1-1)(w_2-1)}{w_1w_2}\cdot\gcd(n, q-1)\\
&+&\prod_{\mbox{\tiny$
\begin{array}{c}
p| m_{w_1w_2,1}\\
p~prime\\
\end{array}$
}}(1+v_p(m_{w_1w_2,1})\frac{p-1}p)\cdot \frac{w_2-1}{w_1w_2}\cdot\gcd(n, q^{w_1}-1) \\
&+&\prod_{\mbox{\tiny$
\begin{array}{c}
p| m_{w_1w_2,2}\\
p~prime\\
\end{array}$
}}(1+v_p(m_{w_1w_2,2})\frac{p-1}p)\cdot \frac{w_1-1}{w_1w_2}\cdot\gcd(n, q^{w_2}-1) \\&+& \prod_{\mbox{\tiny$
\begin{array}{c}
p| m_{w_1w_2}\\
p~prime\\
\end{array}$
}}(1+v_p(m_{w_1w_2})\frac {p-1}p) \cdot\frac{1}{w_1w_2}\cdot\gcd(n, q^{w_1w_2}-1) .\end{eqnarray*}   }
\end{thm}

\begin{proof}

Let $n=w_1^{v_{w_1}(n)}w_2^{v_{w_2}(n)}n_0n_1n_2n_3$, where
$rad(n_0)|(q-1)$,  $rad(n_1)|\frac{q^{w_1}-1}{q-1}$, $rad(n_2)\mid\frac{q^{w_2}-1}{q-1}$, and $\gcd(n_3,(q^{w_1}-1)(q^{w_2}-1))=1$. Let
   $m_{{w_1w_2},0}=\frac{n_0}{\gcd(n_0, q-1)}$, $$m_{{w_1w_2},1}=\frac{w_1^{v_{w_1}(n)}n_0n_1}{\gcd(w_1^{v_{w_1}(n)}n_0n_1, q^{w_1}-1)}=w_1^{v_{w_1}(m_{w_1w_2,1})}m_{{w_1w_2},0}\cdot\frac{n_1}{\gcd(n_1, \frac{q^{w_1}-1}{q-1})},$$ and  $$m_{{w_1w_2},2}=\frac{w_2^{v_{w_2}(n)}n_0n_2}{\gcd(w_2^{v_{w_2}(n)}n_0n_2, q^{w_2}-1)}=w_2^{v_{w_2}(m_{w_1w_2,1})}m_{{w_1w_2},0}\cdot\frac{n_2}{\gcd(n_2, \frac{q^{w_2}-1}{q-1})}.$$  It is easy to verify that $\gcd(n_1,n_2)=1$ and
\begin{eqnarray*}
&&m_{w_1w_2}=\frac{n}{\gcd(n, q^{w_1w_2}-1)}\\
&=&w_2^{v_{w_2}(m_{w_1w_2})}m_{{w_1w_2},1}\cdot\frac{n_2n_3}{\gcd(n_2n_3, \frac{q^{w_1w_2}-1}{q^{w_1}-1})}=w_1^{v_{w_1}(m_{w_1w_2})}m_{{w_1w_2},2}\cdot\frac{n_1n_3}{\gcd(n_1n_3, \frac{q^{w_1w_2}-1}{q^{w_2}-1})}.
\end{eqnarray*}

In (4.2), by  Lemma 2.2, for each divisor $t$ of $m_{w_1w_2}$, the number of irreducible binomials of degree $t$ in $\Bbb F_{q^{w_1w_2}}[x]$ is
$\frac{\varphi(t)}t\cdot\gcd(n, q^{w_1w_2}-1)$ and
 $x^t-\delta^{ul_{w_1w_2}}\in\Bbb F_q[x]$ if and only if $\delta^{ul_{w_1w_2}}\in \Bbb F_q$ if and only if $\frac{q^{w_1w_2}-1}{q-1}| ul_{w_1w_2}$. Moreover, if $x^t-\delta^{ul_{w_1w_2}}\notin\Bbb F_q[x] $, then
 $x^t-\delta^{ul_{w_1w_2}}\in \Bbb F_{q^{w_1}}[x]$  if and only if $\delta^{ul_{w_1w_2}}\in \Bbb F_{q^{w_1}}\verb|\| \Bbb F_{q^{}}$ if and only if $\frac{q^{w_1w_2}-1}{q^{w_1}-1}| ul_{w_1w_2}$ and $\frac{q^{w_1w_2}-1}{q-1}\nmid ul_{w_1w_2}$; $x^t-\delta^{ul_{w_1w_2}}\in\Bbb F_{q^{w_2}}[x]$  if and only if $\delta^{ul_{w_1w_2}}\in \Bbb F_{q^{w_2}}\verb|\| \Bbb F_{q^{}}$ if and only if $\frac{q^{w_1w_2}-1}{q^{w_2}-1}| ul_{w_1w_2}$ and $\frac{q^{w_1w_2}-1}{q-1}\nmid ul_{w_1w_2}$.


In the following, we find these $u$, $1\le u\le \gcd(n, q^{w_1w_2}-1)$ and  $\gcd(u,t)=1$,  such that either $\delta^{ul_{w_1w_2}}\in \Bbb F_q$ or   $\delta^{ul_{w_1w_2}}\in \Bbb F_{q^{w_1}} \verb|\|  \Bbb F_{q}$ or $\delta^{ul_{w_1w_2}}\in \Bbb F_{q^{w_2}} \verb|\|  \Bbb F_{q}$.
We proceed the proof in  four cases.

{\bf Case 1.} either $v_{w_1}(n)=0 $ or $v_{w_1}(n)\geqslant v_{w_1}( q^{w_1w_2}-1)$ and either $v_{w_2}(n)=0 $ or $v_{w_2}(n)\geqslant v_{w_2}( q^{w_1w_2}-1)$.
By Lemmas 2.5 and 4.1,
$$\gcd(\frac{q^{w_1w_2}-1}{q-1}, l_{w_1w_2})=\frac{\frac{q^{w_1w_2}-1}{q^{w_1}-1}}{\gcd(n, \frac{q^{w_1w_2}-1}{q^{w_1}-1})}.$$
Then $x^t-\delta^{ul_{w_1w_2}}\in \Bbb F_q[x]$ is irreducible   if and only if $\gcd(n, \frac{q^{w_1w_2}-1}{q-1})| u$.

Suppose that  $t| m_{{w_1w_2},0}$. Then  $x^t-\delta^{ul_{w_1w_2}}\in\Bbb F_q[x]$ is irreducible if and only if $u=\gcd(n, \frac{q^{w_1w_2}-1}{q-1})u'$ for $1\le u'\le \gcd(n, q-1)$ and $\gcd(t,u')=1$. Hence
there exist $\frac{\varphi(t)}t\cdot
\gcd(n, q-1)$   irreducible polynomials  of degree $t$ in $\Bbb F_{q^{w_1w_2}}[x]$ that are in $\Bbb F_q[x]$. By the way, there is no polynomial of degree $t$ in  $\Bbb F_q[x]$ if $t\nmid m_{w_1w_2,0}$.

By Lemmas 2.5 and 4.1,  $$\gcd(\frac{q^{w_1w_2}-1}{q^{w_1}-1}, l_{w_1w_2})=\frac{\frac{q^{w_1w_2}-1}{q^{w_1}-1}}{\gcd(n, \frac{q^{w_1w_2}-1}{q^{w_1}-1})},$$
  $x^t-\delta^{ul_{w_1w_2}}\notin \Bbb F_q[x]$ and $x^t-\delta^{ul_{w_1w_2}}\in \Bbb F_{q^{w_1}}[x]$ is irreducible    if and only if $\gcd(n, \frac{q^{w_1w_2}-1}{q^{w_1}-1})| u$ and $\gcd(n, \frac{q^{w_1w_2}-1}{q^{}-1})\nmid u$ if and only if $u=\gcd(n, \frac{q^{w_1w_2}-1}{q^{w_1}-1})u_1$ for $1\le u_1\le \gcd(n, q^{w_1}-1)$, $\gcd(u_1,t)=1$ and $\frac{q^{w_1}-1}{q-1}\nmid u_1l_{w_1}$.
Hence
there exist $\frac{\varphi(t)}{w_1t}\cdot
(\gcd(n, q^{w_1}-1)-\gcd(n,q-1))$   irreducible polynomials  of degree $w_1t$ in $\Bbb F_{q^{w_1w_2}}[x]$ that are in $\Bbb F_{q^{w_1}}[x]$.

By Lemmas 2.5 and 4.1,  $$\gcd(\frac{q^{w_1w_2}-1}{q^{w_2}-1}, l_{w_1w_2})=\frac{\frac{q^{w_1w_2}-1}{q^{w_2}-1}}{\gcd(n, \frac{q^{w_1w_2}-1}{q^{w_2}-1})},$$
  $x^t-\delta^{ul_{w_1w_2}}\notin \Bbb F_q[x]$ and $x^t-\delta^{ul_{w_1w_2}}\in \Bbb F_{q^{w_2}}[x]$  if and only if $\gcd(n, \frac{q^{w_1w_2}-1}{q^{w_2}-1})| u$ and $\gcd(n, \frac{q^{w_1w_2}-1}{q^{}-1})\nmid u$ if and only if $u=\gcd(n, \frac{q^{w_1w_2}-1}{q^{w_2}-1})u_2$ for $1\le u_2\le \gcd(n, q^{w_2}-1)$, $\gcd(u_2,t)=1$ and $\frac{q^{w_2}-1}{q-1}\nmid u_2l_{w_2}$.
Hence
there exist $\frac{\varphi(t)}{w_2t}\cdot
(\gcd(n, q^{w_2}-1)-\gcd(n,q-1))$   irreducible polynomials  of degree $w_2t$ in $\Bbb F_{q^{w_1w_2}}[x]$ that are in $\Bbb F_{q^{w_2}}[x]$. Thus  the number of irreducible polynomials of degree $w_1w_2t$ in $\Bbb F_q[x]$ is
$$\frac {\varphi(t)}{w_1w_2t}\cdot( \gcd(n,q^{w_1w_2}-1)-\gcd(n, q^{w_1}-1)- \gcd(n, q^{w_2}-1)+\gcd(n, q-1)).$$

Suppose that $t| m_{{w_1w_2},1}$ and $t\nmid m_{{w_1w_2},0}$. Then there is no irreducible polynomials of degree $t$ and $w_2t$ over $\Bbb F_q$. Hence there are $$\frac{\varphi(t)}{w_1t}\cdot
\gcd(n, q^{w_1}-1)\mbox{ and }\frac{\varphi(t)}{w_1w_2t}\cdot
(\gcd(n, q^{w_1w_2}-1)-\gcd(n, q^{w_1}-1))$$
irreducible polynomials of degree $w_1t$ and $w_1w_2t$ over $\Bbb F_q$, respectively.

Suppose that $t| m_{{w_1w_2},2}$ and $t\nmid m_{{w_1w_2},0}$. Then there is no irreducible polynomials of degree $t$ and $w_1t$ over $\Bbb F_q$. Hence there are $$\frac{\varphi(t)}{w_2t}\cdot
\gcd(n, q^{w_2}-1)\mbox{ and }\frac{\varphi(t)}{w_1w_2t}\cdot
(\gcd(n, q^{w_1w_2}-1)-\gcd(n, q^{w_2}-1))$$
irreducible polynomials of degree $w_2t$ and $w_1w_2t$ over $\Bbb F_q$, respectively.

Suppose that $t\nmid m_{{w_1w_2},1}$ and $t\nmid m_{{w_1w_2},2}$. Then there is no irreducible polynomials of degree $t$, $w_1t$ and $w_2t$ over $\Bbb F_q$.
Hence there exist  $$\frac {\varphi(t)}{w_1w_2t}\cdot \gcd(n,q^{w_1w_2}-1)$$  irreducible polynomials  of degree $w_1w_2t$  in $\Bbb F_{q^{w_1}}[x]$.




%


Hence  we have the irreducible factorization of $x^n-1$ over $\Bbb F_q$ and the number of divisors  is
\begin{eqnarray*}&& \sum_{\mbox{\tiny$
\begin{array}{c}
t|m_{w_1w_2,0}\\
\end{array}$
}}\frac{\varphi(t)}{w_1w_2t}\cdot (\gcd(n, q^{w_1w_2}-1)+(w_2-1)\gcd(n, q^{w_1}-1+(w_1-1)\gcd(n, q^{w_2}-1))\\
&+&(w_1-1)(w_2-1)\gcd(n,q-1))\\
&+&\sum_{\mbox{\tiny$
\begin{array}{c}
t|m_{w_1w_2,1}\\
t\nmid m_{w_1w_2,0}\\
\end{array}$
}}\frac{\varphi(t)}{w_1w_2t}\cdot (\gcd(n, q^{w_1w_2}-1)+(w_2-1)\gcd(n, q^{w_1}-1))\\&+&\sum_{\mbox{\tiny$
\begin{array}{c}
t|m_{w_1w_2,2}\\
t\nmid m_{w_1w_2,0}\\
\end{array}$
}}\frac{\varphi(t)}{w_1w_2t}\cdot (\gcd(n, q^{w_1w_2}-1)+(w_1-1)\gcd(n, q^{w_2}-1))\\&+&
\sum_{\mbox{\tiny$
\begin{array}{c}
t \nmid m_{w_1w_2,1}\\
t \nmid m_{w_1w_2,2}\\
\end{array}$
}}
\frac{\varphi(t)}{w_1w_2t}\cdot\gcd(n, q^{w_1w_2}-1)\\
&=&\prod_{\mbox{\tiny$
\begin{array}{c}
p| m_{w_1w_2,0}\\
p~prime\\
\end{array}$
}}(1+v_p(m_{w_1w_2,0})\frac{p-1}p)\cdot \frac{(w_1-1)(w_2-1)}{w_1w_2}\cdot\gcd(n, q-1)\\
&+&\prod_{\mbox{\tiny$
\begin{array}{c}
p| m_{w_1w_2,1}\\
p~prime\\
\end{array}$
}}(1+v_p(m_{w_1w_2,1})\frac{p-1}p)\cdot \frac{w_2-1}{w_1w_2}\cdot\gcd(n, q^{w_1}-1) \\&+&\prod_{\mbox{\tiny$
\begin{array}{c}
p| m_{w_1w_2,2}\\
p~prime\\
\end{array}$
}}(1+v_p(m_{w_1w_2,2})\frac{p-1}p)\cdot \frac{w_1-1}{w_1w_2}\cdot\gcd(n, q^{w_2}-1) \\&+&
\prod_{\mbox{\tiny$
\begin{array}{c}
p| m_{w_1w_2}\\
p~prime\\
\end{array}$
}}(1+v_p(m_{w_1w_2})\frac {p-1}p) \cdot\frac{1}{w_1w_2}\cdot\gcd(n, q^{w_1w_2}-1).
\end{eqnarray*}

{\bf Case 2.} either $v_{w_1}(n)=0 $ or $v_{w_1}(n)\geqslant v_{w_1}( q^{w_1w_2}-1)$ and $1\leqslant v_{w_2}(n)<v_{w_2}( q^{w_1w_2}-1)$.
Then
$$
\gcd(\frac{q^{w_1w_2}-1}{q-1}, l_{w_1w_2})=\frac{w_2\cdot\frac{q^{w_1w_2}-1}{q-1}}{\gcd(n, \frac{q^{w_1w_2}-1}{q-1})},\gcd(\frac{q^{w_1w_2}-1}{q^{w_1}-1}, l_{w_1w_2})=\frac{w_2\cdot\frac{q^{w_1w_2}-1}{q^{w_1}-1}}{\gcd(n, \frac{q^{w_1w_2}-1}{q^{w_1}-1})},$$
and $$\mbox{  ~}\gcd(\frac{q^{w_1w_2}-1}{q^{w_2}-1}, l_{w_1w_2})=\frac{\frac{q^{w_1w_2}-1}{q^{w_2}-1}}{\gcd(n, \frac{q^{w_1w_2}-1}{q^{w_2}-1})}.$$


{\bf Case 3.} either $v_{w_2}(n)=0 $ or $v_{w_2}(n)\geqslant v_{w_1}( q^{w_1w_2}-1)$ and $1\leqslant v_{w_1}(n)<v_{w_1}( q^{w_1w_2}-1)$.
Then
$$\gcd(\frac{q^{w_1w_2}-1}{q-1}, l_{w_1w_2})=\frac{w_1\cdot\frac{q^{w_1w_2}-1}{q-1}}{\gcd(n, \frac{q^{w_1w_2}-1}{q-1})},\gcd(\frac{q^{w_1w_2}-1}{q^{w_1}-1}, l_{w_1w_2})=\frac{\frac{q^{w_1w_2}-1}{q^{w_1}-1}}{\gcd(n, \frac{q^{w_1w_2}-1}{q^{w_1}-1})},$$
 and $$\mbox{ ~}\gcd(\frac{q^{w_1w_2}-1}{q^{w_2}-1}, l_{w_1w_2})=\frac{w_1\cdot\frac{q^{w_1w_2}-1}{q^{w_2}-1}}{\gcd(n, \frac{q^{w_1w_2}-1}{q^{w_2}-1})}.$$


{\bf Case 4.}  $1\leqslant v_{w_1}(n)<v_{w_1}( q^{w_1w_2}-1)$ and $1\leqslant v_{w_2}(n)<v_{w_2}( q^{w_1w_2}-1)$.
Then
$$\gcd(\frac{q^{w_1w_2}-1}{q-1}, l_{w_1w_2})=\frac{w_1w_2\cdot\frac{q^{w_1w_2}-1}{q-1}}{\gcd(n, \frac{q^{w_1w_2}-1}{q-1})},\gcd(\frac{q^{w_1w_2}-1}{q^{w_1}-1}, l_{w_1w_2})=\frac{w_2\cdot\frac{q^{w_1w_2}-1}{q^{w_1}-1}}{\gcd(n, \frac{q^{w_1w_2}-1}{q^{w_1}-1})},$$
and $$\mbox{  ~}\gcd(\frac{q^{w_1w_2}-1}{q^{w_2}-1}, l_{w_1w_2})=\frac{w_1\cdot\frac{q^{w_1w_2}-1}{q^{w_2}-1}}{\gcd(n, \frac{q^{w_1w_2}-1}{q^{w_2}-1})}.$$

Similarly, we can prove them.
\end{proof}

 \begin{exa}{\rm Suppose that $w_1=3, w_2=5, q=2$ and $n=151$. Then these positive integers satisfy the condition in Theorem $4.2$, i.e. $ord_{rad(n)}(q)=15$ and either $q\not\equiv3 \pmod 4$ or $8\nmid n$. By calculate directly,   $m_{15}=m_{15,1}=m_{15,2}=m_{15,3}=1$.By Theorem 4.2, the number of irreducible factor  of $x^{151}-1$ over $\Bbb F_2$ is $\frac{151}{15}+\frac 2 {15}+\frac 4{15}+\frac8{15}=11$, which is confirmed by Magma.  }

 \end{exa}

Second, we consider the case: $q\equiv 3\pmod 4$ and $8|n$.  By $w_1w_2$ an odd integer,  $q^{w_1w_2} \equiv3 \pmod 4$  if and only if $q \equiv3 \pmod 4$.

Let $r=\min\{v_2(n/2),v_2(q^{w_1w_2}+1)\}=\min\{v_2(n/2),v_2(q+1)\}$, $m_{2{{w_1w_2}}}=\frac{n}{\gcd(n, q^{2{{w_1w_2}}}-1)}$, $l_{2{{w_1w_2}}}=\frac{q^{2{{w_1w_2}}}-1}{\gcd(n, q^{2{{w_1w_2}}}-1)}$, $l_{{w_1w_2}}=\frac{q^{{w_1w_2}}-1}{\gcd(n, q^{{w_1w_2}}-1)}$, $\mu$ a  generator of $\Bbb F_{q^{2{{w_1w_2}}}}^*$, and   $\delta=\mu^{q^{{w_1w_2}}+1}$.
Then the factorization of $x^{n} -1$ into irreducible factors in $\Bbb F_{q^{{w_1w_2}}} [x ]$ is

\begin{equation}
\prod_{\mbox{\tiny$
\begin{array}{c}
t|m_{2{{w_1w_2}}}\\
t~odd\\
\end{array}$
}} \prod_{\mbox{\tiny$
\begin{array}{c}
1\leqslant v\leqslant \gcd(n,q^{{w_1w_2}}-1)\\
\gcd(v,t)=1\\
\end{array}$
}}(x^{t}-\delta^{vl_{{{w_1w_2}}}})\cdot\prod\limits_{t\mid m_{2{{w_1w_2}}}} \prod\limits_{u\in \mathcal{R}_{t}} (x^{2t}-(\mu^{ul_{2{{w_1w_2}}}}+\mu^{q^{{w_1w_2}}ul_{2{{w_1w_2}}}})x^{t}+\delta^{ul_{2{{w_1w_2}}}})
\end{equation}
where $$\mathcal{R}_t=\bigg\{u\in \Bbb N: \begin{array}{l}
1\le u\le \gcd(n,q^{2{{w_1w_2}}}-1),\gcd(u,t)=1,
 \\ 2^{r}\nmid u,  u=min\{u, q^{{w_1w_2}}u\}_{\gcd(n,q^{2{{w_1w_2}}}-1)}\end{array}\bigg\}.$$

  Next we will explicitly give the irreducible factorization of $x^{n}-1$ over $\Bbb F_{q}$.

\begin{thm} {\rm Suppose that ${w_1, w_2}$ are two odd distinct primes, $ord_{rad(n)}(q)={w_1w_2}$, $q^{} \equiv 3\pmod 4$,  and $8| n$. Let $n=w_1^{v_{w_1}(n)}w_2^{v_{w_2}(n)}n_0n_1n_2n_3$, where
$rad(n_0)|(q-1)$,  $rad(n_1)|\frac{q^{w_1}-1}{q-1}$, $rad(n_2)|\frac{q^{w_2}-1}{q-1}$, and $\gcd(n_3,(q^{w_1}-1)(q^{w_2}-1))=1$. Let
   $m_{{w_1w_2},0}=\frac{n_0}{\gcd(n_0, q-1)}$, $m_{{w_1w_2},1}=\frac{w_1^{v_{w_1}(n)}n_0n_1}{\gcd(w_1^{v_{w_1}(n)}n_0n_1, q^{w_1}-1)}$,  $m_{{w_1w_2},2}=\frac{w_2^{v_{w_2}(n)}n_0n_2}{\gcd(w_2^{v_{w_2}(n)}n_0n_2, q^{w_2}-1)}$,  $m_{2w_1w_2}=\frac{n}{\gcd(n,q^{2w_1w_2}-1)}$, $m_{w_1w_2}=\frac{n}{\gcd(n,q^{w_1w_2}-1)}$,
  $l_{2w_1w_2}=\frac{q^{2w_1w_2}-1}{\gcd(n,q^{2w_1w_2}-1)}$,  $l_{w_1w_2}=\frac{q^{w_1w_2}-1}{\gcd(n,q^{w_1w_2}-1)}$, $l_{w_1}=\frac{q^{w_1}-1}{\gcd(n,q^{w_1}-1)}$, $l_{w_2}=\frac{q^{w_2}-1}{\gcd(n,q^{w_2}-1)}$, and $l_{1}=\frac{q-1}{\gcd(n, q-1)}$. Let $r=min\{v_{2}(\frac n 2),v_{2}(q^{w_1w_2}+1)\}$ and $\mu$ be a  generator of $\Bbb F_{q^{2{{w_1w_2}}}}^*$ satisfying  $\delta=\mu^{q^{{w_1w_2}}+1}$,
  $\pi_1=\mu^{\frac{q^{2w_1w_2}-1}{q^{w_1}-1}}$, $\pi_2=\mu^{\frac{q^{2w_1w_2}-1}{q^{w_2}-1}}$, $\alpha=\mu^{\frac{q^{2{w_1w_2}}-1}{q^2-1}}$, and $\theta=\mu^{\frac{q^{2w_1w_2}-1}{q-1}}$. Then

$(1)$ The irreducible factorization of $x^{n}-1$ over $\Bbb F_{q}$ is
\begin{eqnarray*}&&\prod_{\mbox{\tiny$
\begin{array}{c}
 t|m_{w_1w_2,0}\\t~odd
\end{array}$
}}
 \prod_{\mbox{\tiny$
\begin{array}{c}
1\le v\le \gcd(n, q-1)\\
\gcd(v, t)=1
\end{array}$
}}(x^{t}-\theta^{vl_{1}})\cdot\prod_{\mbox{\tiny$
\begin{array}{c}
t|m_{w_1w_2,1}\\ t~odd
\end{array}$
}}\prod_{\mbox{\tiny$
\begin{array}{c}
v_1\in \mathcal{S}^{(1)}_{t}
\end{array}$
}}
\prod\limits_{k=0}^{w_1-1}(x^{t}-\pi_1^{q^kv_1l_{w_1}})\\
&\cdot&\prod_{\mbox{\tiny$
\begin{array}{c}
t|m_{w_1w_2,2}\\ t~odd
\end{array}$
}}\prod_{\mbox{\tiny$
\begin{array}{c}
v_2\in \mathcal{S}^{(2)}_{t}
\end{array}$
}}
\prod\limits_{k=0}^{w_2-1}(x^{t}-\pi_2^{q^kv_2l_{w_2}})\cdot\prod_{\mbox{\tiny$
\begin{array}{c}
t|m_{2w_1w_2}\\ t~odd
\end{array}$
}}\prod_{\mbox{\tiny$
\begin{array}{c}
v_3\in \mathcal{S}^{(3)}_{t}
\end{array}$
}}
\prod\limits_{k=0}^{w_1w_2-1} (x^{t}-\delta^{q^kv_3l_{w_1w_2}})\\
&\cdot&\prod_{\mbox{\tiny$
\begin{array}{c}
t| m_{w_1w_2,0}\\u_1\in \mathcal{R}^{(1)}_t
\end{array}$
}}
(x^{2t}-(\alpha^{u_1l_{2}}+\alpha^{qu_1l_{2}})x^{t}+\theta^{u_1l_{2}}) \cdot\prod_{\mbox{\tiny$
\begin{array}{c}
t| m_{2w_1w_2}\\
u_2\in\mathcal{ R}^{(2)}_{t}
\end{array}$
}}\prod\limits_{k=0}^{2w_1w_2-1}(x^t-\mu^{q^ku_2l_{2w_1w_2}}),
\end{eqnarray*}
where~ $$\mathcal{S}^{(1)}_{t}=\bigg\{v_1\in \Bbb N: \begin{array}{l}
1\le v_1\le \gcd(n,q^{w_1}-1),\gcd(v_1,t)=1,
 \\ \frac{q^{w_1}-1}{q-1}\nmid v_1l_{w_1}, v_1=\min\{v_1, qv_1, \cdots, q^{w_1-1}v_1\}_{\gcd(n, q^{w_1}-1)}\end{array}\bigg\},$$
 $$\mathcal{S}^{(2)}_{t}=\bigg\{v_2\in \Bbb N: \begin{array}{l}
1\le v_2\le \gcd(n,q^{w_2}-1),\gcd(v_2,t)=1,
 \\ \frac{q^{w_2}-1}{q-1}\nmid v_2l_{w_2}, v_2=\min\{v_2, qv_2, \cdots, q^{w_2-1}v_2\}_{\gcd(n, q^{w_2}-1)}\end{array}\bigg\},$$
 $$\mathcal{S}^{(3)}_{t}=\bigg\{v_3\in \Bbb N: \begin{array}{l}
1\le v_3\le \gcd(n,q^{w_1w_2}-1),\gcd(v_3,t)=1,
 \\ \frac{q^{w_1w_2}-1}{q^{w_1}-1}\nmid v_3l_{w_1w_2}, \frac{q^{w_1w_2}-1}{q^{w_2}-1}\nmid v_3l_{w_1w_2},\\v_3=\min\{v_3, qv_3, \cdots, q^{w_1w_2-1}v_3\}_{\gcd(n, q^{w_1w_2}-1)}\end{array}\bigg\},$$
$$\mathcal{R}^{(1)}_t=\bigg\{u_1\in \Bbb N: \begin{array}{l}
1\le u_1\le 2^{r}\gcd(n,q^{}-1),\gcd(u_1,t)=1,
 \\ 2^r\nmid u_1, u_1=\min\{u_1,q^{w_1w_2}u_1\}_{\gcd(n, q^2-1)}\end{array}\bigg\},$$
 and~$$\mathcal{R}_t^{(2)}=\bigg\{u_2\in \Bbb N: \begin{array}{l}
1\le u_2\le \gcd(n,q^{2w_1w_2}-1),\gcd(u_2,t)=1, \frac{q^{2{{w_1w_2}}}-1}{q^2-1}\nmid u_2l_{2{{w_1w_2}}},
 \\  2^{r}\nmid u_2,  u_2=\min\{u_2, qu_2,\cdots,  q^{2{w_1w_2}-1}u_2\}_{\gcd(n,q^{2{w_1w_2}}-1)}\end{array}\bigg\}.$$


$(2)$
The number of irreducible factors of $x^n-1$ in $\Bbb F_q[x]$ is
\begin{eqnarray*}
&&\prod_{\mbox{\tiny$
\begin{array}{c}
p|m_{w_1w_2,0}\\
p~odd~prime
\end{array}$
}}(1+v_p(m_{2w_1w_2})\frac {p-1}p)\cdot\frac{B}{2w_1w_2}\cdot \gcd(n, q-1) \\
&+&\prod_{\mbox{\tiny$
\begin{array}{c}
p|m_{w_1w_2,1}\\
p ~odd ~prime
\end{array}$
}}
(1+v_p(m_{2w_1w_2})\frac {p-1}p)\cdot \frac{w_2-1}{w_1w_2}\cdot\gcd(n,q^{w_1}-1)\\&+&\prod_{\mbox{\tiny$
\begin{array}{c}
p|m_{w_1w_2,2}\\
p ~odd~prime
\end{array}$
}}
(1+v_p(m_{2w_1w_2})\frac {p-1}p)\cdot \frac{w_1-1}{w_1w_2}\cdot\gcd(n,q^{w_2}-1)\\
&+&\prod_{\mbox{\tiny$
\begin{array}{c}
p|m_{2w_1w_2}\\
p~odd ~prime
\end{array}$
}}
(1+v_p(m_{2w_1w_2})\frac {p-1}p)\cdot\frac{2^r+1+v_2(m_{2w_1w_2})2^{r-1}}{2w_1w_2} \cdot\gcd(n,q^{w_1w_2}-1),
\end{eqnarray*}
where $B=2^{r-1}(2+v_2(m_{2w_1w_2}))+w_1w_2-2w_1-2w_2+3$. }
\end{thm}

\begin{proof}


 Since $rad(n)|(q^{{w_1w_2}}-1)$ and $w_1w_2$ is an odd integer, $\gcd(n/2, q^{{{w_1w_2}}}+1)=\gcd(n/2, q+1)=2^{r}$, where $r=\min\{v_2(n/2), v_2(q^{{w_1w_2}}+1)\}=\min\{v_2(n/2),v_2(q+1)\}$. Then $$l_{2{{w_1w_2}}}=\frac{q^{2{{w_1w_2}}}-1}{\gcd(n,q^{2{{w_1w_2}}}-1)}=\frac{q^{{{w_1w_2}}}-1}{\gcd(n,q^{{{w_1w_2}}}-1)}\cdot \frac{q^{{{w_1w_2}}}+1}{\gcd(n/2,q^{{{w_1w_2}}}+1)}=\frac{q^{{{w_1w_2}}}+1}{2^{r}}l_{{{w_1w_2}}}$$ and $$ ~m_{2{{w_1w_2}}}=\frac{n}{\gcd(n,q^{2{w^2}}-1)}=\frac{n}{\gcd(n/2,q^{{w^2}}+1)\gcd(n,q^{{{w_1w_2}}}-1)}=\frac{m_{{{w_1w_2}}}}{2^{r}}.$$


Firstly, we investigate the product in (4.3):

$$\prod_{\mbox{\tiny$
\begin{array}{c}
t| m_{2{{w_1w_2}}}\\
t~odd\\
\end{array}$
}} \prod_{\mbox{\tiny$
\begin{array}{c}
1\leqslant v\leqslant \gcd(n,q^{{w_1w_2}}-1)\\
\gcd(v,t)=1\\
\end{array}$
}}(x^{t}-\delta^{vl_{{{w_1w_2}}}}).$$
Similar to Theorem 4.2, we obtain all irreducible polynomials of degrees $t$, $w_1t$, $w_2t$ and $w_1w_2t$  as follows.

Suppose that  $t| m_{{w_1w_2},0}$. The numbers of irreducible polynomials of degrees $t$, $w_1t$, $w_2t$ and $w_1w_2t$  are $$\frac{\varphi(t)}t\cdot
\gcd(n, q-1),\frac{\varphi(t)}{w_1t}\cdot
(\gcd(n, q^{w_1}-1)-\gcd(n,q-1)),\frac{\varphi(t)}{w_2t}\cdot
(\gcd(n, q^{w_2}-1)-\gcd(n,q-1)),$$
and $$\mbox{ } \frac {\varphi(t)}{w_1w_2t}\cdot( \gcd(n,q^{w_1w_2}-1)-\gcd(n, q^{w_1}-1)- \gcd(n, q^{w_2}-1)+\gcd(n, q-1)).$$

Suppose that $t| m_{{w_1w_2},1}$ and $t\nmid m_{{w_1w_2},0}$. Then there is no irreducible polynomials of degrees $t$ and $w_2t$ over $\Bbb F_q$. Hence there are $$\frac{\varphi(t)}{w_1t}\cdot
\gcd(n, q^{w_1}-1)\mbox{ and }\frac{\varphi(t)}{w_1w_2t}\cdot
(\gcd(n, q^{w_1w_2}-1)-\gcd(n, q^{w_1}-1))$$
irreducible polynomials of degrees $w_1t$ and $w_1w_2t$ over $\Bbb F_q$, respectively.

Suppose that $t| m_{{w_1w_2},2}$ and $t\nmid m_{{w_1w_2},0}$. Then there is no irreducible polynomials of degrees $t$ and $w_1t$ over $\Bbb F_q$. Hence there are $$\frac{\varphi(t)}{w_2t}\cdot
\gcd(n, q^{w_2}-1)\mbox{ and }\frac{\varphi(t)}{w_1w_2t}\cdot
(\gcd(n, q^{w_1w_2}-1)-\gcd(n, q^{w_2}-1))$$
irreducible polynomials of degree $w_2t$ and $w_1w_2t$ over $\Bbb F_q$, respectively.

Suppose that $t\nmid m_{{w_1w_2},1}$ and $t\nmid m_{{w_1w_2},2}$. Then there is no irreducible polynomials of degrees $t$, $w_1t$ and $w_2t$ over $\Bbb F_q$.
Hence there exist  $\frac {\varphi(t)}{w_1w_2t}\cdot \gcd(n,q^{w_1w_2}-1)$  irreducible polynomials  of degree $w_1w_2t$  in $\Bbb F_{q^{w_1}}[x]$.

Secondly, we investigate the product in (4.3):
\begin{equation*}\prod\limits_{t\mid m_{2{{w_1w_2}}}} \prod\limits_{u\in \mathcal{R}_{t}} (x^{2t}-(\mu^{ul_{2{{w_1w_2}}}}+\mu^{q^{{w_1w_2}}ul_{2{{w_1w_2}}}})x^{t}+\delta^{ul_{2{{w_1w_2}}}}). \end{equation*}
By Lemma 4.1, we prove this part by four cases:

{\bf Case 1.} either $v_{w_1}(n)=0 $ or $v_{w_1}(n)\geqslant v_{w_1}( q^{w_1w_2}-1)$ and either $v_{w_2}(n)=0 $ or $v_{w_2}(n)\geqslant v_{w_2}( q^{w_1w_2}-1)$.
 By Lemmas 2.5 and 4.1,
$$\gcd(\frac{q^{2w_1w_2}-1}{q^2-1}, l_{2w_1w_2})=\frac{\frac{q^{2w_1w_2}-1}{q^2-1}}{\gcd(n, \frac{q^{w_1w_2}-1}{q-1})}.$$
Note that the irreducible polynomial  $x^{2t}-(\mu^{ul_{2{w_1w_2}}}+\mu^{q^{{{w_1w_2}}}ul_{2{{w_1w_2}}}})x^{t}+\delta^{ul_{2{{w_1w_2}}}}=(x^t-\mu^{ul_{2{{w_1w_2}}}})(x^t-\mu^{q^{{w_1w_2}}ul_{2{{w_1w_2}}}})\in \Bbb F_{q^{w_1w_2}}[x]$ and $x^t-\mu^{ul_{2{{w_1w_2}}}}\notin \Bbb F_{q}[x]$. Then $x^{2t}-(\mu^{ul_{2{w_1w_2}}}+\mu^{q^{{{w_1w_2}}}ul_{2{{w_1w_2}}}})x^{t}+\delta^{ul_{2{{w_1w_2}}}} \in \Bbb F_{q}[x]$
 if and only if $x^t-\mu^{ul_{2{{w_1w_2}}}}\in \Bbb F_{q^2}[x]$ if and only if $\frac{q^{2{{w_1w_2}}}-1}{q^2-1}|ul_{2{{w_1w_2}}}$ and $\frac{q^{2{{w_1w_2}}}-1}{q-1}\nmid ul_{2{{w_1w_2}}}$  if and only if $u=\gcd(n_{}, \frac{q^{{w_1w_2}}-1}{q-1})u_1$ for $1\le u_1\le 2^{r}\gcd(n, q-1)$, $\gcd(t,u_1)=1$ and $2^r\nmid u_1$, i.e.  $x^{2t}-(\mu^{ul_{2{{w_1w_2}}}}+\mu^{q^{{{w_1w_2}}}ul_{2{{w_1w_2}}}})x^{t}+\delta^{ul_{2{{w_1w_2}}}}=x^{2t}-(\alpha^{u_1l_2}+\alpha^{qu_1l_2})x^t+\theta^{u_1l_2}\in \Bbb F_q[x]$ is irreducible.

In Lemma 2.3, the number of irreducible polynomials of degree $2t$  depends on the parity of $t$. In the following, we divide the proof into two cases.

Suppose that  $t|m_{w_1w_2, 0}$.  If $t$ is odd,  there exist $\frac{\varphi(t)}{2t}\cdot(2^{r_{}}-1)\cdot\gcd(n,q-1)$ irreducible polynomials of degree $2t$ in $\Bbb F_q[x]$.
  On the other hand, each irreducible polynomial  $x^{2t}-(\mu^{ul_{2{w_1w_2}}}+\mu^{q^{{w_1w_2}}ul_{2{w_1w_2}}})x^{t}+\delta^{ul_{2{w_1w_2}}}\in \Bbb F_{q^{w_1w_2}}[x]$ with $\frac{q^{2{{w_1w_2}}}-1}{q^2-1}\nmid ul_{2{{w_1w_2}}}$, there are ${w_1w_2}$  conjugate irreducible polynomials  in $\Bbb F_{q^{w_1w_2}}[x]$ such that their product generates an irreducible polynomial in $\Bbb F_q[x]$, i.e. $\prod_{k=0}^{2{w_1w_2}-1}(x^t-\pi^{q^kul_{2{w_1w_2}}})\in \Bbb F_q[x]$ is irreducible. From Lemma 2.3, the number of irreducible polynomials of degree $2t$ over $\Bbb F_{q^{w_1w_2}}$ in (4.3) is $\frac{\varphi(t)}{2t}\cdot(2^{r_{}}-1)\cdot\gcd(n,q^{{w_1w_2}}-1)$. Hence
   the number of irreducible polynomials of degree $2{w_1w_2}t$ over  $\Bbb F_{q}$ is $$\frac{\varphi(t)}{2{w_1w_2}t}\cdot(2^{r_{}}-1)\cdot(\gcd(n,q^{{w_1w_2}}-1)-\gcd(n,q-1)).$$
 If  $t$ is even,  then $u_1$ must be an odd integer and the condition $2^r\nmid u_1$ automatically holds. Hence, the numbers of irreducible polynomials of degree $2t$ and $2{w_1w_2}t$ over  $\Bbb F_{q}$ are $$\frac{\varphi(t)}{2t}\cdot2^{r_{}}\cdot\gcd(n,q-1)\mbox{ and }\frac{\varphi(t)}{2{w_1w_2}t}\cdot2^{r_{}}\cdot(\gcd(n,q^{{w_1w_2}}-1)-\gcd(n,q-1)).$$

Suppose that  $t\nmid m_{{w_1w_2},0}$. Then there is no polynomial of degree $2t$ in  $\Bbb F_q[x]$. If $t$ is even,
 then the number of irreducible polynomials of degree $2{w_1w_2}t$ in $\Bbb F_q[x]$ is
$$\frac{\varphi(t)}{2{w_1w_2}t}\cdot2^{r}\cdot\gcd(n,q^{{w_1w_2}}-1).$$ If $t$ is an odd integer, the number of irreducible polynomials of degree $2{w_1w_2}t$ in $\Bbb F_q[x]$ is
$$\frac{\varphi(t)}{2{w_1w_2}t}\cdot(2^{r_{}}-1)\cdot\gcd(n,q^{{w_1w_2}}-1).$$

 Hence we have the irreducible factorization of $x^n-1$ over $\Bbb F_q$ in  (1). The number of irreducible factors of $x^n-1$ in $\Bbb F_q[x]$ is
\begin{eqnarray*}&&\sum_{\mbox{\tiny$
\begin{array}{c}
t|m_{w_1w_2,0}\\
t~odd
\end{array}$
}}
\frac{\varphi(t)}{2w_1w_2t}\cdot((2^r+1)\gcd(n, q^{w_1w_2}-1)+2(w_2-1)\gcd(n, q^{w_1}-1)\\
&+&2(w_1-1)\gcd(n, q^{w_2}-1)+(2(w_1-1)(w_2-1)+(2^r-1)(w_1w_2-1))\gcd(n,q-1))\\
&+&\sum_{\mbox{\tiny$
\begin{array}{c}
t|m_{w_1w_2,0}\\
t~even
\end{array}$
}}
\frac{\varphi(t)}{2w_1w_2t}\cdot2^r\cdot(\gcd(n, q^{w_1w_2}-1)+(w_1w_2-1)\gcd(n, q-1))\\
&+&\sum_{\mbox{\tiny$
\begin{array}{c}
t|m_{w_1w_2,1}\\t\nmid m_{w_1w_2,0}\\t~odd
\end{array}$
}}
\frac{\varphi(t)}{2w_1w_2t}\cdot((2^r+1)\gcd(n, q^{w_1w_2}-1)+2(w_2-1)\gcd(n,q^{w_1}-1))\\&+&\sum_{\mbox{\tiny$
\begin{array}{c}
t|m_{w_1w_2,1}\\t\nmid m_{w_1w_2,0}\\t~even
\end{array}$
}}
\frac{\varphi(t)}{2w_1w_2t}\cdot2^r\cdot\gcd(n, q^{w_1w_2}-1)\\&+&\sum_{\mbox{\tiny$
\begin{array}{c}
t|m_{w_1w_2,2}\\t\nmid m_{w_1w_2,0}\\t~odd
\end{array}$
}}
\frac{\varphi(t)}{2w_1w_2t}\cdot((2^r+1)\gcd(n, q^{w_1w_2}-1)+2(w_1-1)\gcd(n,q^{w_2}-1))\\
&+&\sum_{\mbox{\tiny$
\begin{array}{c}
t|m_{w_1w_2,2}\\t\nmid m_{w_1w_2,0}\\t~even
\end{array}$
}}
\frac{\varphi(t)}{2w_1w_2t}\cdot2^r\cdot\gcd(n, q^{w_1w_2}-1)\\
&+&\sum_{\mbox{\tiny$
\begin{array}{c}
t\nmid m_{w_1w_2,1}\\
t\nmid m_{w_1w_2,2}\\t~odd
\end{array}$
}}
\frac{\varphi(t)}{2w_1w_2t}\cdot(2^r+1)\gcd(n, q^{w_1w_2}-1)+\sum_{\mbox{\tiny$
\begin{array}{c}
t\nmid m_{w_1w_2,1}\\
t\nmid m_{w_1w_2,2}\\t~even
\end{array}$
}}
\frac{\varphi(t)}{2w_1w_2t}\cdot2^r\cdot\gcd(n, q^{w_1w_2}-1)\\
&=&\prod_{\mbox{\tiny$
\begin{array}{c}
p|m_{w_1w_2,0}\\
p~odd~prime
\end{array}$
}}(1+v_p(m_{2w_1w_2})\frac {p-1}p)\cdot\frac{B}{2w_1w_2}\cdot \gcd(n, q-1) \\
&+&\prod_{\mbox{\tiny$
\begin{array}{c}
p|m_{w_1w_2,1}\\
p ~odd ~prime
\end{array}$
}}
(1+v_p(m_{2w_1w_2})\frac {p-1}p)\cdot \frac{w_2-1}{w_1w_2}\cdot\gcd(n,q^{w_1}-1)\\&+&\prod_{\mbox{\tiny$
\begin{array}{c}
p|m_{w_1w_2,2}\\
p ~odd~prime
\end{array}$
}}
(1+v_p(m_{2w_1w_2})\frac {p-1}p)\cdot \frac{w_1-1}{w_1w_2}\cdot\gcd(n,q^{w_2}-1)\\
&+&\prod_{\mbox{\tiny$
\begin{array}{c}
p|m_{2w_1w_2}\\
p~odd ~prime
\end{array}$
}}
(1+v_p(m_{2w_1w_2})\frac {p-1}p)\cdot\frac{2^r+1+v_2(m_{2w_1w_2})2^{r-1}}{2w_1w_2} \cdot\gcd(n,q^{w_1w_2}-1),
\end{eqnarray*}
where $B=2^{r-1}(2+v_2(m_{2w_1w_2}))+w_1w_2-2w_1-2w_2+3$.

{\bf Case 2.} either $v_{w_1}(n)=0 $ or $v_{w_1}(n)\geqslant v_{w_1}( q^{w_1w_2}-1)$ and $1\leqslant v_{w_2}(n)<v_{w_2}( q^{w_1w_2}-1)$.
By Lemmas 2.5 and 4.1,

$$\gcd(\frac{q^{2w_1w_2}-1}{q^2-1}, l_{2w_1w_2})=\frac{w_2\cdot\frac{q^{2w_1w_2}-1}{q^2-1}}{\gcd(n, \frac{q^{w_1w_2}-1}{q-1})}.$$


{\bf Case 3.} either $v_{w_2}(n)=0 $ or $v_{w_2}(n)\geqslant v_{w_1}( q^{w_1w_2}-1)$ and $1\leqslant v_{w_1}(n)<v_{w_1}( q^{w_1w_2}-1)$.
By Lemmas 2.5 and 4.1,

$$\gcd(\frac{q^{2w_1w_2}-1}{q^2-1}, l_{2w_1w_2})=\frac{w_1\cdot\frac{q^{2w_1w_2}-1}{q^2-1}}{\gcd(n, \frac{q^{w_1w_2}-1}{q-1})}.$$


{\bf Case 4.}  $1\leqslant v_{w_1}(n)<v_{w_1}( q^{w_1w_2}-1)$ and $1\leqslant v_{w_2}(n)<v_{w_2}( q^{w_1w_2}-1)$.
By Lemmas 2.5 and 4.1,
$$\gcd(\frac{q^{2w_1w_2}-1}{q^2-1}, l_{2w_1w_2})=\frac{w_1w_2\cdot\frac{q^{2w_1w_2}-1}{q^2-1}}{\gcd(n, \frac{q^{w_1w_2}-1}{q-1})}.$$

Similarly, we can prove them.
 \end{proof}

\begin{exa}{\rm Suppose that $q=3, w_1=3,w_2=5$ and $n=36488$.  It is easy to check that these integers satisfying the condition in Theorem $4.4$, i.e. $ord_{rad(n)}(q)=15$ and  $q\equiv3 \pmod 4$ and $8| n$.
By  calculate directly,  $r=2$, $m_{2w_1w_2}=1 $, and $m_{w_1w_2,1}=m_{w_1w_2,2}=m_{w_1w_2,3}=4$.  Hence the number of irreducible factors of $x^{36488}-1$ over $\Bbb F_{3}$ is
$\frac{5\times 4561}{15}+\frac{8}{15}+\frac{4}{15}+\frac{4\times 14+2}{15}=1525$ by Theorem $4.4$. The result is confirmed by Magma.  }


\end{exa}

\subsection{$w_1w_2$  is an even number }$~$

As $w_1,w_2$ are two distinct primes, in this subsection, we  always assume that $w_1=2$, $w_2=w$ is an  odd integer and $\ord_{rad(n)}(q)=2w$. It is noteworthy that in this subsection, we just have to consider one case as $q^{2w}\not\equiv3\pmod4$ for each $q$.

Similar to Lemma 4.1,  the following lemma holds.
\begin{lem}{\rm Let $w$ be an odd prime. Then



$(1)$ If either $v_{2}(n)=0 $ or $v_{2}(n)\geqslant v_{2}( q^{2w}-1)$ and $v_{w}(n)=0 $ or $v_{w}(n)\geqslant v_{w}( q^{2w}-1)$, then
$\gcd(n, q^{2w}-1)=\gcd(n, q^{}-1)\gcd(n, \frac{q^{2w}-1}{q-1})=\gcd(n, q^{2}-1)\gcd(n, \frac{q^{2w}-1}{q^{2}-1})=\gcd(n, q^{w}-1)\gcd(n, \frac{q^{2w}-1}{q^{w}-1}).$

$(2)$ If $v_{2}(n)=0 $ or $v_{2}(n)\geqslant v_{2}( q^{2w}-1)$ and $1\leqslant v_{w}(n)<v_{w}( q^{2w}-1)$, then
$
\gcd(n, q^{2w}-1)=\gcd(n/w, q^{}-1)\gcd(n, \frac{q^{2w}-1}{q-1})=\gcd(n/w, q^{2}-1)\gcd(n, \frac{q^{2w}-1}{q^{2}-1})=\gcd(n, q^{w}-1)\gcd(n, \frac{q^{2w}-1}{q^{w}-1}).
$

$(3)$ If $v_{w}(n)=0 $ or $v_{w}(n)\geqslant v_{w}( q^{2w}-1)$ and $1\leqslant v_{2}(n)<v_{2}( q^{2w}-1)$, then
\begin{eqnarray*}
\gcd(n, q^{2w}-1)&=&\gcd(n, q^{2}-1)\gcd(n, \frac{q^{2w}-1}{q^{2}-1})\\
&=&\left\{\begin{array}{ll}
\gcd(n/2, q^{}-1)\gcd(n, \frac{q^{2w}-1}{q-1})&\mbox{ if $q\equiv1\pmod4$,}\\
\gcd(n/4, q^{}-1)\gcd(n, \frac{q^{2w}-1}{q-1})&\mbox{ if $q\equiv3\pmod4$,}\\
\end{array}\right.\\
&=&\left\{\begin{array}{ll}
\gcd(n/2, q^{w}-1)\gcd(n, \frac{q^{2w}-1}{q^{w}-1})&\mbox{ if $q\equiv1\pmod4$,}\\
\gcd(n/4, q^{w}-1)\gcd(n, \frac{q^{2w}-1}{q^{w}-1})&\mbox{ if $q\equiv3\pmod4$.}\\
\end{array}\right.\\
\end{eqnarray*}
$(4)$ If $1\leqslant v_{2}(n)<v_{2}( q^{2w}-1)$  and $1\leqslant v_{w}(n)<v_{w}( q^{2w}-1)$, then
\begin{eqnarray*}
\gcd(n, q^{2w}-1)&=&\gcd(n/w, q^{2}-1)\gcd(n, \frac{q^{2w}-1}{q^{2}-1})\\& =&\left\{\begin{array}{ll}
\gcd(n/2w, q^{}-1)\gcd(n, \frac{q^{2w}-1}{q-1}),&\mbox{ if $q\equiv1\pmod4$,}\\
\gcd(n/4w, q^{}-1)\gcd(n, \frac{q^{2w}-1}{q-1}),&\mbox{ if $q\equiv3\pmod4$,}\\
\end{array}\right.\\
&=&\left\{\begin{array}{ll}
\gcd(n/2, q^{w}-1)\gcd(n, \frac{q^{2w}-1}{q^{w}-1}),&\mbox{ if $q\equiv1\pmod4$,}\\
\gcd(n/4, q^{w}-1)\gcd(n, \frac{q^{2w}-1}{q^{w}-1}),&\mbox{ if $q\equiv3\pmod4$.}\\
\end{array}\right.\\
\end{eqnarray*}


}
\end{lem}
By Lemmas 2.5, 4.6 and the proof Theorem 4.2,  the following result holds.

\begin{thm} {\rm Suppose that $\ord_{rad(n)}(q)=2w$, and $w$ is an odd  prime.
 Let $n=2^{v_{2}(n)}w^{v_{w}(n)}n_0n_1n_2n_3$, where
$rad(n_0)|(q-1)$,  $rad(n_1)|(q+1)$, $rad(n_2)|\frac{q^{4}-1}{q^2-1}$, and $\gcd(n_3, $ $(q^2-1)(q^{w}-1))=1$. Let
   $m_{{2w},0}=\frac{n_0}{\gcd(n_0, q-1)}$, $m_{{2w},1}=\frac{2^{v_{2}(n)}n_0n_1}{\gcd(2^{v_{2}(n)}n_0n_1, q^{w_1}-1)}$,  $m_{{2w},2}=\frac{w^{v_{w}(n)}n_0n_2}{\gcd(w^{v_{w}(n)}n_0n_2, q^{w_2}-1)}$, $m_{2w}=\frac{n}{\gcd(n,q^{2w}-1)}$, $l_{2w}=\frac{q^{2w}-1}{\gcd(n,q^{2w}-1)}$, $l_{w}=\frac{q^{w}-1}{\gcd(n,q^{w}-1)}$, $l_{2}=\frac{q^{2}-1}{\gcd(n,q^{2}-1)}$, $l_{1}=\frac{q-1}{\gcd(n, q-1)}$,  $\delta$ a  generator of $\Bbb F^{\ast}_{q^{2w}}$ satisfying $\pi=\delta^{\frac{q^{2w}-1}{q^{w}-1}}$,   $\alpha=\delta^{\frac{q^{2w}-1}{q^{2}-1}}$, and $\theta=\delta^{\frac{q^{2w}-1}{q^{}-1}}$. Then

$(1)$ The irreducible factorization of  $x^{n}-1$  over $\Bbb F_{q}$  is
\begin{eqnarray*}&&\prod_{\mbox{\tiny$
\begin{array}{c}
 t|m_{2w,0}\\
\end{array}$
}}
 \prod_{\mbox{\tiny$
\begin{array}{c}
1\le v\le \gcd(n, q-1)\\
\gcd(v, t)=1
\end{array}$
}}(x^{t}-\theta^{vl_{1}})\\
&\cdot&\prod_{\mbox{\tiny$
\begin{array}{c}
t|m_{2w,1}\\
\end{array}$
}}\prod_{\mbox{\tiny$
\begin{array}{c}
v_1\in \mathcal{S}^{(1)}_{t}
\end{array}$
}}
(x^{2t}-(\alpha^{v_1l_{2}}+\alpha^{qv_1l_{2}})x^t+\theta^{v_1l_2})\\
&\cdot&\prod_{\mbox{\tiny$
\begin{array}{c}
t|m_{2w,2}\\
\end{array}$
}}\prod_{\mbox{\tiny$
\begin{array}{c}
v_2\in \mathcal{S}^{(2)}_{t}
\end{array}$
}}
\prod\limits_{k=0}^{w-1}(x^{t}-\pi^{q^kv_2l_{w}})\cdot\prod_{\mbox{\tiny$
\begin{array}{c}
t|m_{2w}\\
\end{array}$
}}\prod_{\mbox{\tiny$
\begin{array}{c}
v_3\in \mathcal{S}^{(3)}_{t}
\end{array}$
}}
\prod\limits_{k=0}^{2w-1} (x^{t}-\delta^{q^kv_3l_{2w}}),
\end{eqnarray*}
where $$\mbox{  }\mathcal{S}^{(1)}_{t}=\bigg\{v_1\in \Bbb N: \begin{array}{l}
1\le v_1\le \gcd(n,q^{2}-1),\gcd(v_1,t)=1,
 \\ (1+q)\nmid v_1l_{2}, v_1=\min\{v_1, qv_1\}_{\gcd(n, q^{2}-1)}\end{array}\bigg\},$$
 $$\mathcal{S}^{(2)}_{t}=\bigg\{v_2\in \Bbb N: \begin{array}{l}
1\le v_2\le \gcd(n,q^{w}-1),\gcd(v_2,t)=1,
 \\ \frac{q^{w}-1}{q-1}\nmid v_2l_{w}, v_2=\min\{v_2, qv_2, \cdots, q^{w-1}v_2\}_{\gcd(n, q^{w}-1)}\end{array}\bigg\},$$
and $$\mbox{  }\mathcal{S}^{(3)}_{t}=\bigg\{v_3\in \Bbb N: \begin{array}{l}
1\le v_3\le \gcd(n,q^{2w}-1),\gcd(v_3,t)=1,
  \frac{q^{2w}-1}{q^{2}-1}\nmid v_3l_{2w}, \\ \frac{q^{2w}-1}{q^{w}-1}\nmid v_3l_{2w},  v_3=\min\{v_3, qv_3, \cdots, q^{2w-1}v_3\}_{\gcd(n, q^{2w}-1)}\end{array}\bigg\}.$$
$(2)$ The number of irreducible factors is
\begin{eqnarray*}&&\prod_{\mbox{\tiny$
\begin{array}{c}
p| m_{2w,0}\\
p~prime\\
\end{array}$
}}(1+v_p(m_{2w,0})\frac{p-1}p)\cdot \frac{w-1}{2w}\cdot\gcd(n, q-1) \\
&+&\prod_{\mbox{\tiny$
\begin{array}{c}
p| m_{2w,1}\\
p~prime\\
\end{array}$
}}(1+v_p(m_{2w,1})\frac{p-1}p)\cdot \frac{w-1}{2w}\cdot\gcd(n, q^{2}-1) \\
&+&\prod_{\mbox{\tiny$
\begin{array}{c}
p| m_{2w,2}\\
p~prime\\
\end{array}$
}}(1+v_p(m_{2w,2})\frac{p-1}p)\cdot \frac{1}{2w}\cdot\gcd(n, q^{w}-1) \\&+&\prod_{\mbox{\tiny$
\begin{array}{c}
p| m_{2w}\\
p~prime\\
\end{array}$
}}(1+v_p(m_{2w})\frac {p-1}p) \cdot\frac{1}{2w}\cdot\gcd(n, q^{2w}-1) .\end{eqnarray*}
}
\end{thm}

\section{Concluding remarks}
Suppose that $w_1,w_2$ are two prime numbers, and $ord_{rad(n)}(q)=w_1w_2$. In this paper,   we explicitly factorized  $x^{n}-1$ over $\Bbb F_{q}$ in the two cases of $w_1=w_2$ and  $w_1\neq w_2$.

Suppose that $ord_{rad(n)}(q)=d=p_1^{\alpha_1}\cdots p_l^{\alpha_l}$, where $p_1,\cdots,p_l$ are distinct prime numbers and $\alpha_i>0$, $1\le i\le l$. Assume that all positive divisors of $d$ are  $d_1,\cdots,d_k$.
Since $rad(n)|(q^d-1)$ and $q^d-1=\prod_{d'|d}\Phi_{d'}(q)$, $n$ is rewritten  as follows.
$$n=p_1^{v_{p_1}(n)}\cdots p_l^{v_{p_l}(n)}\cdot n_{d_1}n_{d_2}\cdots n_{d_k}\cdot\overline{n},$$
where $\gcd(\overline{n},(q^{d/p_1}-1)\cdots(q^{d/p_l}-1))=1$, $rad(n_{d_1})|\Phi_{d_1}(q), \cdots$, and $ rad(n_{d_k})|\Phi_{d_k}(q)$,
After a complicated calculation, we can also obtain the factorization of $x^n-1$ over $\Bbb F_q$ and the number of irreducible factors, which is similar to those theorems above.

\bigskip
\section*{Acknowledgments} The second  author was supported by the National Natural Science Foundation of China (No. 61772015) and the Foundation of Science and Technology on Information Assurance Laboratory (No. KJ-17-010). The authors wish to thank the anonymous reviewers and the Associate Editor for their very helpful comments that improved the presentation and quality of this paper.





\end{document}